\documentclass{article}
\usepackage{graphicx} 
\usepackage{xcolor} 

\title{Tight Bounds for Gaussian Mean Estimation under Personalized Differential Privacy}
\date{}
\usepackage{booktabs}
\usepackage[authoryear]{natbib}

\usepackage{colortbl}
\usepackage{tablefootnote}
\usepackage{multirow}
\usepackage{graphicx}
\usepackage{tabularx}
\usepackage{float} 
\usepackage{enumitem}
\usepackage{amsmath, booktabs, threeparttable}
\usepackage{subcaption}
\usepackage{fullpage}
\usepackage{enumitem} 
\usepackage{amsfonts}

\usepackage{xcolor}

\usepackage[letterpaper,top=3cm,bottom=3cm,left=3cm,right=3cm,marginparwidth=1.75cm]{geometry}

\usepackage{amsmath,amsthm,amssymb,amsfonts}
\usepackage{graphicx}
\usepackage[colorlinks=true, allcolors=blue]{hyperref}
\usepackage[ruled,vlined,linesnumbered]{algorithm2e}

\newtheorem{theorem}{Theorem}
\newtheorem{lemma}{Lemma}

\theoremstyle{definition}
\newtheorem{definition}{Definition}

\theoremstyle{remark}

\allowdisplaybreaks

\author{Wei Dong$^{\ast}$, and Li Ge$^{\ast}$ \\[0.23em]
        {\fontsize{10pt}{10pt}\selectfont Nanyang Technological University, Singapore} \\[0.23em]
        {\fontsize{10pt}{10pt}\selectfont \texttt{\{wei\_dong, li040\}@ntu.edu.sg}} \\[0.23em]
        {\fontsize{10pt}{10pt}\selectfont $^{\ast}$Author ordering is alphabetical.}}
\begin{document}
\maketitle

\begin{abstract}
We study mean estimation for Gaussian distributions under \textit{personalized differential privacy} (PDP), where each record has its own privacy budget. PDP is commonly considered in two variants: \textit{bounded} and \textit{unbounded} PDP. In bounded PDP, the privacy budgets are public and neighboring datasets differ by replacing one record. In unbounded PDP, neighboring datasets differ by adding or removing a record; consequently, an algorithm must additionally protect participation information, making both the dataset size and the privacy profile sensitive. Existing works have only studied mean estimation over bounded distributions under bounded PDP. Different from mean estimation for distributions with bounded range, where each element can be treated equally and we only need to consider the privacy diversity of elements, the challenge for Gaussian is that, elements can have very different contributions due to the unbounded support. we need to jointly consider the privacy information and the data values. Such a problem becomes even more challenging under unbounded PDP, where the privacy information is protected and the way to compute the weights becomes unclear. In this paper, we address these challenges by proposing optimal Gaussian mean estimators under both bounded and unbounded PDP, where in each setting we first derive lower bounds for both problems, following PDP mean estimators with the algorithmic upper bounds matching the corresponding lower bounds up to logarithmic factors.

\end{abstract}

\clearpage

\tableofcontents

\clearpage

\section{Introduction}

Estimating the mean of a distribution from independent samples is among the most fundamental tasks in modern machine learning and statistics. Given a dataset $x_1, x_2, \hdots, x_n$ of i.i.d. samples drawn from an unknown distribution $\mathcal{P}$, the goal is to estimate the mean $\mathbb{E}_{x\sim \mathcal{P}}[x]$ using these observations. In the classical setting, this problem has been extensively studied, and a wide range of estimators with strong statistical guarantees are now well understood. However, in many real-world applications, the data may contain sensitive information, and directly releasing the vanilla sampled mean (or related statistics) can lead to privacy leakage, enabling membership inference or fingerprinting attacks \citep{PODS_2003,homer2008resolving,bun2014fingerprinting,dwork2015robust}. These concerns motivate the study of \textit{private mean estimation}, where one seeks accurate estimation while simultaneously protecting individuals' privacy.

One standard formalism for privacy protection is \textit{differential privacy} (DP) \citep{dwork2006our,dwork2006calibrating,dwork2014algorithmic}, which requires that the output distributions of an algorithm on any two neighboring datasets that differ in one sample differ by at most $e^{\varepsilon}$. Here, $\varepsilon$ is called the privacy budget, which is used to control the privacy loss: smaller budget leads to less privacy loss. W
Depending on how ``neighboring'' datasets are defined, two common DP models are considered: bounded DP and unbounded DP. In bounded DP, neighboring datasets differ in one record (i.e., one record is replaced). In unbounded DP, neighbors differ by the addition or removal of a single record. The unbounded model provides stronger protection as it can also obscure an individual’s participation in the dataset, not merely the value of their record.
Mean estimation under DP has been studied extensively, resulting in a rich body of work spanning various distributional settings, under both bounded DP\citep{KV18,colt_2019,nips_2019,clipping_2,clipping_3,Amin_nips_2019,Alabi_stoc,aden2021sample,AOS_2021,pmlr-v162-asi22b,colt_2023_fast,add_remove_one_mean_estimation} and unbounded DP \citep{add_remove_one_mean_estimation,2025_statistical_estimation}. However, standard DP enforces a uniform privacy level across all individuals, i.e., adopting the same privacy budget $\varepsilon$ across all individuals, which may be overly restrictive when they have heterogeneous privacy preferences.

\textit{Personalized differential privacy} (PDP) \citep{Graham,partitioning_based_PDP,PDP_platform} addresses this limitation by allowing each individual to specify their own privacy budget $\varepsilon_i$. Informally, changing user $i$'s record should not change the output distribution by more than a multiplicative factor $e^{\varepsilon_i}$. By accommodating heterogeneous privacy preferences, PDP offers a more flexible privacy model than standard DP.
Similar as DP, PDP also comes in two variants: bounded PDP and unbounded PDP.
\begin{itemize}[leftmargin=*, itemsep=2pt]
\item \textbf{Bounded PDP} Each individual’s privacy budget is treated as public. Neighboring datasets differ in the content of exactly one element. This model primarily protects data values, while privacy parameters remain publicly known.
\item \textbf{Unbounded PDP} Neighboring datasets differ by adding or removing one record. In this setting, both data values and their existence must be protected, which makes the privacy information itself sensitive (since it is tied to whether an individual is included).
\end{itemize}





Compared to bounded PDP, which protects only data values, unbounded PDP additionally protects participation: that is, whether an individual appears in the dataset. Consequently, both the dataset size and the privacy parameters of participating individuals become input-dependent and cannot be publicly released, making unbounded PDP substantially more challenging. For mean estimation in particular, existing work has so far focused on the bounded PDP setting \citep{PDP_2023,PDP_platform}. Unbounded PDP has been studied for other data-analysis tasks, e.g., \citep{Graham,partitioning_based_PDP}.

In this work, we study PDP mean estimation for Gaussian distributions, arguably the most fundamental distributional family in statistical estimation \citep{Alabi_stoc,aden2021sample,ashtiani2022private}, in both the bounded and unbounded settings. Prior work \citep{PDP_2023} investigates mean estimation under bounded PDP for bounded distributions, specifically, a distribution with input domain $[-1/2, 1/2]$. In their work, they formalize the problem into a weighted mean estimation problem. Since the distribution is bounded, each value's contribution is also well bounded and can be treated equally. Therefore, they only need to care about the diversity of privacy budgets among elements. To further address such privacy heterogeneity, they assign each value a weight. With the noise scale fixed, smaller weights lead to a stronger privacy budget, while also leading to an extra sampling error. By adjusting the weights based on (public) privacy information, they achieve both personalized privacy protection and, at the same time, an optimal tradeoff between noise scale and sampling error led from the rescaled weights, which further leads to an optimal total error.


In contrast, the Gaussian case introduces an extra key difficulty: Gaussian distributions have \emph{unbounded support} and thus different elements may have very different contributions. To achieve a good or even optimal error, we not only need to consider the privacy information of elements, but also need to consider the data values. As a result, the problem cannot be simply solved with a weighted mean estimation. This problem becomes even more challenging under unbounded PDP since, in their weighted mean estimator, the weights are computed directly from the public privacy information, which becomes unavailable in unbounded PDP. This will paper tackles the above challenges, establishing tight bounds for both settings. Our main results are as follows.

\paragraph{The bounded PDP setting}
We begin with mean estimation under the bounded PDP model, where the (ascendingly sorted) privacy vector
$\boldsymbol{\varepsilon} = (\varepsilon_1,\ldots,\varepsilon_n)$, containing all individuals' privacy budgets, is treated as public information.

\subparagraph{Lower bound}
Our first result is a minimax lower bound for bounded PDP mean estimation:
\begin{theorem}[Informal, see Theorem \ref{theorem:lower_bound}]
For any $\boldsymbol{\varepsilon}$-PDP mean estimator $\mathcal{M}$ operating on a dataset $\mathcal{N}(\mu, \sigma^2)^n$, with probability at least $1/4$, its estimate deviates from $\mu$ by at least
\begin{align*}
    \Omega\left(
    \max_{k = 1, \cdots, n}
    \frac{\sigma}{\sum_{i=1}^{k}\varepsilon_i + \sqrt{n-k}}
    \right).
\end{align*}
\end{theorem}

This lower bound has a mixing formulation of DP-like error and sampling error and is parameterized by the shape of the Gaussian distribution, i.e., $\sigma$.

\subparagraph{Upper bound}
Our main algorithmic contribution in the bounded PDP model is an estimator that matches the lower bound up to logarithmic factors:
\begin{theorem}[Informal, see Theorem \ref{theorem:upper_bound_bounded}]
Given $\beta > 0$, $\boldsymbol{\varepsilon} \leq 1^n$, there is an $\boldsymbol{\varepsilon}$-PDP estimator operating on a dataset $D = \mathcal{N}(\mu, \sigma^2)^n$ such that, if $n = \tilde{\Omega}(1/\varepsilon_1)$,
with probability at least $1-\beta$,
\begin{align*}
\vert \mathcal{M}(D)-\mu\vert
\leq
\tilde{O}\Big(\max_{k = 1, \cdots, n}
    \frac{\sigma}{\sum_{i=1}^{k}\varepsilon_i + \sqrt{n-k}}\Big),
\end{align*}
where $\tilde{\Omega}(\cdot)$ and $\tilde{O}(\cdot)$ suppresses polylogarithmic factors in $n$, $1/\beta$, and the relevant privacy parameters.
\end{theorem}

The estimator follows a clip-then-estimate roadmap. The key technical challenge is identifying a $\tilde{O}(\sigma)$-width region around $\mu$ as the clipping range while respecting heterogeneous privacy budgets. To address this, we introduce a new primitive called \emph{diffusion}, which achieves privacy amplification in bounded PDP with \emph{non-uniform} rates. Diffusion works by sampling elements at different rates and replacing unsampled elements with a placeholder $\perp$ rather than removing them, preserving the change-one neighboring structure while enabling personalized privacy protection. Using diffusion, we design a PDP-compliant range estimator that balances diffusion-induced sampling error against privacy-induced noise. After clipping to the estimated range, we invoke a weighted PDP mean estimator to get the final result.

\paragraph{The Unbounded PDP setting}
We next study mean estimation under the unbounded PDP model. In this model, the privacy budget of each individual is specified by a public privacy function $\mathcal{E}$, i.e. $\mathcal{E}(u)$, but the privacy vector for records in a given dataset becomes private, since it reveals which records are present. Consequently, the mean estimator must safeguard both the data values and the privacy profile.

\subparagraph{Lower bound}
We also start from the lower bound, and first show that the lower bound is still governed by a rate of a similar form as in the bounded setting, except that it becomes \emph{privacy-specific}:

\begin{theorem}[Informal, see Theorem \ref{theorem:lower_bound_add/remove-one}]
For any $\mathcal{E}$-PDP mean estimator $\mathcal{M}$ operating on a dataset $D = \mathcal{N}(\mu, \sigma^2)^n$ with a privacy vector $\boldsymbol{\varepsilon}_D$ that contains the privacy budgets for all records in $D$, with probability at least $1/4$, its estimate deviates from $\mu$ by at least
\begin{align*}
    \Omega\left(
    \max_{k = 1, \cdots, n}
    \frac{\sigma}{\sum_{i=1}^{k}{\varepsilon_D}_i + \sqrt{n-k}}
    \right).
\end{align*}
\end{theorem}
We derive the lower bound by establishing an output-input distance relationship for distributions after they are processed by an $\mathcal{E}$-PDP mechanism. This technical step allows us to translate separations between candidate input distributions into separations between the corresponding output distributions, and thus obtain a Le Cam-style minimax lower bound. As a result, the unbounded PDP lower bound admits a formulation that closely parallels the bounded PDP case.

\subparagraph{Upper bound}
Next, we also show an algorithmic upper bound that matches this privacy-specific lower bound up to logarithmic factors:
\begin{theorem}[Informal, see Theorem \ref{theorem:upper_bound_add/remove-one}]
Given $\beta > 0$, there is an $\mathcal{E}$-PDP mean estimator $\mathcal{M}$ operating on a dataset $D = \mathcal{N}(\mu, \sigma^2)^n$ with a privacy vector $\boldsymbol{\varepsilon}_D$ that contains all individuals' privacy specification $\mathcal{E}(u)$ such that, if $n = \tilde{\Omega}(1/\varepsilon_{\min}^{3/2}(D))$, with probability at least $1-\beta$,
\begin{align*}
\vert \mathcal{M}(D)-\mu\vert
\leq
\tilde{O}\Big(\max_{k = 1, \cdots, n}
    \frac{\sigma}{\sum_{i=1}^{k}\varepsilon_i + \sqrt{n-k}}\Big),
\end{align*}
where $\tilde{\Omega}(\cdot)$ and $\tilde{O}(\cdot)$ suppresses polylogarithmic factors in $n$, $1/\beta$, and the relevant privacy parameters.
\end{theorem}
One can notice that there is a stricter, though still mild, requirement on the size of the dataset, yet the algorithmic upper bound remains optimal up to logarithmic factors. Our unbounded-PDP estimator is built via a reduction to the bounded-PDP problem. The reduction must simultaneously: (i) avoid any privacy downgrade, (ii) keep the transformation's impact on the dataset small enough to retain utility, and (iii) be distance-preserving so that add/remove-one neighbors map to bounded change-one neighbors at constant distance. We accomplish this by privately \emph{coarsening} the privacy profile using an exponentially-spaced partition of $[\varepsilon_{\min},\varepsilon_{\max}]$. This produces a public (sanitized) proxy privacy vector and a corresponding \emph{shrunk} dataset on which we can safely run the bounded-PDP estimator.

\subsection{Related Work}
Mean estimation under standard differential privacy has been studied extensively. Prior work includes studies of private mean estimation for Gaussian distributions, both single and multivariate settings \citep{KV18,colt_2019,nips_2019,clipping_2,AOS_2021,clipping_1,Private_estimation_with_public_data,Alabi_stoc}, as well as for distributions with bounded moments or heavy tails \citep{clipping_3,heavy_tailed,efficient,colt_2023_fast,Average_case_averages}. Several works have also focused on the related problem of private PCA, covariance estimation \citep{Dwork_stoc_2014,Hardt_nips_2014,nips_2019}, and estimating higher-order moments \citep{aden2021sample,ashtiani2022private,kothari2022private,Universal}. In addition, private mean estimation for arbitrary bounded distributions has been studied in \citep{bun2014fingerprinting,dwork2015robust}.

Several notions of heterogeneous privacy have been proposed under the personalized differential privacy / heterogeneous differential privacy \citep{partitioning_based_PDP,Graham,zhang2019}. General purpose approaches include grouping/partitioning individuals with similar privacy levels \citep{partitioning_based_PDP}, applying conservative privacy guarantees \citep{Alaggan}, or using sampling-based ideas that down-sample individuals with stricter privacy requirements and then run standard DP algorithms on the sampled dataset \citep{Graham}. While broadly applicable, these methods are typically not instance-optimal for mean estimation and may become overly conservative when some individuals are effectively ``public'' (very loose privacy demands).

The most closely related line of work studies mean estimation in the bounded input domain under PDP, where the privacy information is public. A result in this setting is the affine/weighted estimator that computes a saturated version of $\boldsymbol{\varepsilon}$ and adds appropriately calibrated Laplace noise, achieving instance-wise minimax optimality over a bounded data domain \citep{PDP_2023}. This work also reveals a ``saturation'' phenomenon: once the most stringent individuals determine the effective noise level, relaxing privacy for other individuals may not further improve accuracy. Related perspectives appear in privacy-aware data acquisition and mechanism design, where the goal is to elicit data from agents with heterogeneous privacy costs, and the resulting optimal policies can exhibit similar saturation effects \citep{PDP_platform,Privacy_Aware_Agents}.

\section{Overview of techniques}

In this section, we give an overview on our developed techniques to arrive at our results.

\subsection{Mean Estimation in Bounded PDP}
We first study mean estimation over Gaussian distribution under the bounded PDP setting. Similar to the traditional DP setting, the central challenge is that the support of the Gaussian distribution may span a very large range, or even be unbounded, leading to a large or potentially infinite sensitivity. In the traditional DP literature, existing solutions \citep{clipping_1,clipping_2,clipping_3,Universal} follow a common strategy: first estimate a concentrated range of the underlying distribution, i.e., $[\mu - \tilde{O}(\sigma), \mu + \tilde{O}(\sigma)]$, and then clip all data points to this estimated range. Once clipped, the data are in a bounded domain, which allows mean estimation mechanisms designed for the bounded domain to be applied. Our bounded PDP mean estimator follows the same high-level roadmap, where we first estimate a PDP range, and then we clip all data points accordingly and then invoke the existing mean estimation mechanism to complete the estimation.

\paragraph{Lower bound}
The minimax lower bound of mean estimation is established by applying the Le Cam's method to the family of Gaussian distribution with variance $\sigma^2$. A crucial step is to relate the distance between the output distribution of the PDP mechanisms to that of the input distribution. One can imagine that this step also introduce PDP parameters into the lower bound since the output distribution is processed by a $\boldsymbol{\varepsilon}$-PDP mechanism. Here, we use a lemma in \citep{PDP_2023} to achieve such a relation and yield a lower bound parameterized by $\sigma$.

\paragraph{Achieving bounded PDP through diffusing} Our unbounded PDP solution relies on a basic primitive called \emph{diffusing} to achieve personalized privacy protection. The high-level idea is to apply a pre-processing step that diversifies the level of privacy protection across individuals. 
One may observe that sampling, more precisely, the privacy amplification effect induced by sampling, can achieve such a target. However, this approach is not applicable in the bounded PDP setting. Specifically, achieving personalized privacy protection requires a non-uniform sampling scheme, and to the best of our knowledge, no existing work has investigated the bounded PDP setting with non-uniform sampling. Fundamentally, the difficulty stems from that, for any two change-one neighbors $D$ and $D'$, after the sampling, we cannot find a valid coupling between their sampled results such that these two sampled results are still change-one neighbors.
To address this issue, we develop a new technique, termed \emph{diffusing}. The idea is straightforward; we sample a subset of the dataset while for those unsampled elements, instead of removing them, we put a placeholder value $\perp$, which maintains the neighbors. The diffusion can be shown to yield a privacy amplification effect under the bounded PDP setting.

\paragraph{PDP range estimation}
With the help of diffusion, one straightforward idea to estimate a PDP range is first generate the diffused dataset and then use existing DP range estimators, like \citep{Universal}. However, there are two main challenges: 1) How to choose appropriate diffusing rates? Note that the selection of these rates is directly related to the error of the mean estimator; there is a tradeoff between sampling error from diffusion and the DP error induced by the latter DP range estimator.
Therefore, the key is to identify diffusing rates that both achieve personalized privacy protection while achieving a favorable trade-off between sampling error and privacy-induced noise. 2) How to adapt existing DP range estimators to operate on the diffused dataset? Existing DP range estimators are designed for datasets with a public size and do not account for placeholder elements. After diffusion, however, the dataset size itself becomes sensitive information, and placeholder values are introduced into the dataset. Consequently, additional processing is required to properly handle both the presence of placeholders and the private data size.
To address these two challenges, we proceed as follows. First, we carefully design our diffusing rates, where our insight is to preserve PDP while matching the lower bound established later.
Then, we manipulate the range estimation framework of~\citep{Universal} to accommodate both placeholder elements and the sensitive data size.

\paragraph{Final error derivation}
With the estimated range in hand, we clip all data points to this range and perform mean estimation using the existing mechanism proposed in~\citep{PDP_2023}. To analyze the estimation error, we decompose the total error into three components in a standard manner: clipping error, sampling error, and the Laplace noise. Among these components, bounding the clipping error requires additional care, as it admits a more intricate structure and directly using the worst case bound yields sub-optimal result. A key observation is that not all outliers contribute equally to the clipping error. By leveraging a concentration bound tailored to the two-stage sampling scheme, we obtain a refined bound on the clipping error. The resulting error bound matches the corresponding lower bound up to logarithmic factors.

\subsection{Mean Estimation in Unbounded PDP}
As mentioned in the Introduction, a key distinction between the bounded PDP setting and the unbounded PDP setting is that, in the latter, the privacy information itself is private. This feature masks the presence or absence of any individual in the dataset and therefore provides stronger privacy protection. As a result, performing mean estimation under the unbounded PDP setting is inherently more challenging than under the bounded setting.

To address the challenge, we reduce the problem of unbounded PDP mean estimation to the bounded PDP problem so that we can invoke our developed techniques in the bounded PDP setting. We first explore the formulation of the corresponding lower bound under the unbounded PDP setting. Then, we show how this problem can be reduced to the bounded PDP problem that we have already addressed.

\paragraph{Privacy-specific lower bound}
In the unbounded PDP setting, privacy budgets themselves are private and may differ across database instances. A straightforward way to establish a minimax lower bound is to take the worst case over all possible privacy vectors. However, such a uniform bound can be overly pessimistic and largely uninformative. Instead, we focus on deriving a lower bound that is tailored to the specific privacy vector of a given dataset, which yields a more informative characterization of the fundamental limit.
To obtain such a lower bound in the unbounded PDP setting, as in the bounded PDP case, we require a relation that connects the distance between two output distributions to the distance between their corresponding input distributions. Moreover, this relation must be able to inject privacy parameters that are specific to the dataset into the lower bound. To this end, we first establish a lemma that relates these two distances. This lemma plays an analogous role to Lemma~\ref{lemma:PDP_2023_TV} in the bounded PDP setting, but explicitly incorporates dataset-specific privacy parameters. Leveraging this relation, we then apply Le Cam’s method to derive a privacy-specific lower bound in the unbounded PDP setting.

\paragraph{Unbounded PDP Mean Estimation}
For mean estimation under the unbounded PDP setting, our main idea is to reduce the problem to a bounded PDP instance, enabling the use of techniques developed for bounded PDP. This reduction faces two fundamental challenges: the privacy vector itself is private and dataset-dependent, and the neighboring relation changes from change-one to insert/delete-one.
To address these challenges, we carefully construct a mapping that transforms the original dataset into a shrunk dataset with a public privacy vector while preserving PDP. This mapping must simultaneously satisfy several requirements: it must preserve PDP, avoid any privacy downgrade, control the size of the transformed dataset to prevent excessive utility loss, and remain distance-preserving so that neighboring datasets are mapped to bounded PDP datasets at a small distance.
A key difficulty is how to balance privacy and utility when we try to obtain the mapped privacy vector. A straightforward idea is to uniformly partition the privacy domain into several intervals and maps each privacy budget to the left end-point of the interval it belongs. However, such a strategy either incurs excessive information loss when the partition is too fine, or collapses the privacy budgets when the partition is too coarse. To overcome this issue, we partition the privacy domain into exponentially expanding intervals. This approach enables us to obtain a truncated, privacy-preserving approximation of the original privacy vector that retains sufficient precision while effectively controlling the number of deletions..
Based on this approximated privacy vector, we construct a shrunk dataset via random sampling and privacy upgrading. Crucially, this transformation is distance-preserving: neighboring datasets under the unbounded PDP relation are mapped to datasets that differ by only a constant number of elements under the bounded PDP relation. This enables us to apply our bounded PDP mean estimation algorithm upon the shrunk dataset. To establish the optimality of our approach, note that, since the upper and lower bounds have been shown to match on the shrunk dataset, showing that the lower bounds on the shrunk and the original datasets are matching suffices to show that the upper and lower bounds match under the unbounded PDP setting.

\section{Preliminaries}
\label{section:preliminaries}
This section contains notations and essential definitions that will be used throughout this paper.

\subsection{Notations}

The dataset is defined as $D = \{x_1, x_2, \ldots, x_n\} \in \mathbb{R}^n$.  Without loss of generality, we assume $D$ is ordered, i.e., $x_i\leq x_j$ for any $i<j$. Given a dataset $D$, we use $x_{\min}(D) = \min_i{x_i}$ and $x_{\max}(D) = \max_i{x_i}$ to denote the \textit{minimum} and \textit{maximum} of $D$. Besides, we define \textit{radius} as the maximum absolute value in $D$, i.e., $\texttt{Radius}(D) := \max_i |x_{i}(D)|$. The \textit{range} of $D$ is defined as the minimal interval that covers all values of $D$, i.e., $\texttt{Range}(D) := [x_{\min}(D), x_{\max}(D)]$, and its \textit{width} is defined as $\omega(D) = x_{\max}(D) - x_{\min}(D)$. Furthermore, $n(D)$ denotes the number of elements in $D$. $\mathcal{D}$ denotes the dataset domain and $\mathcal{Y}$ denotes the output domain of mechanisms. 

Given a distribution $\mathcal{P}$, we use $X \in \mathcal{P}$ to denote a random variable $X$ drawn from $\mathcal{P}$. $\mu$ and $\sigma^2$ denote the mean and variance of $\mathcal{P}$, respectively. 
For any $D$, $\mu(D)$ and $\sigma(D)$ denote the operation of taking the empirical expectation and standard deviation over the data in $D$, respectively. $\Pr(A)$ denotes the probability of event $A$.
In this paper, we study the mean estimation over a Gaussian distribution $\mathcal{N}(\mu,\sigma^2)$: we have $D \sim \mathcal{P}^n$, $\mathcal{P} = \mathcal{N}(\mu,\sigma^2)$, and our target is to estimate the parameter $\mu$ from $D$.

\subsection{Differential Privacy}

\textit{Differential Privacy} (DP) aims to protect the privacy of individual records. By introducing carefully calibrated randomness into query results, it obscures the contribution of any single record, thereby preventing its disclosure. The detailed definition is as follows:

\begin{definition}[Differential Privacy~\citep{dwork2014algorithmic}]
    For any $\varepsilon \geq 0$, a mechanism $\mathcal{M} : \mathcal{D} \to \mathcal{Y}$ is said to be $\varepsilon$-DP if, for all neighboring datasets $D, D' \in \mathcal{D}$, and for all measurable subsets $K \subseteq \mathcal{Y}$, the following holds:
    \begin{equation*}
        \Pr(\mathcal{M}(D) \in Y) \leq e^{\varepsilon} \cdot \Pr(\mathcal{M}(D') \in Y).
    \end{equation*}
\end{definition}

Here, the neighboring datasets $D$ and $D'$, denoted $D \sim D'$, differ by exactly one record; that is, one can be obtained from the other by modifying a single record. Depending on how this modification is defined, we have two main DP policies: the \emph{bounded} DP setting and the \emph{unbounded} DP setting. 
Under the bounded DP, $D'$ is obtained by changing a single record in $D$.
This is commonly referred to as the \textit{change-one} neighboring relation.
In this setting, the data size is fixed, In this setting, the dataset size remains fixed across all datasets, which is the origin of the term of ``bounded''.
In contrast, datasets may have unbounded size in the unbounded DP setting, where $D$ and $D'$ are defined as neighbors if $D \subset D'$ and $\lvert D' \rvert - \lvert D \rvert = 1$ (or vice versa). This neighboring relationship known as the \textit{add/remove-one} neighboring relation.
 
Intuitively, the unbounded DP setting provides stronger privacy protection: it not only hides the content information but even the existence information of the elements. More precisely, any change-one neighbors can be regarded as the add/remove-one neighbors with two-distance by performing one deletion following an insertion. However, the reverse conversion does not hold. Furthermore, in the bounded DP setting, the data size $n(D)$ is public, while it becomes data-dependent under the unbounded DP setting and should be protected. 

DP (both the bounded DP setting and the unbounded DP setting) enjoys some properties, which are listed below:

\begin{lemma}[Post-Processing]
\label{lemma:post-processing_DP}
    If a mechanism $\mathcal{M}: \mathcal{D} \to \mathcal{Y}$ preserves $\varepsilon$-DP, then for any random mechanism $\mathcal{M}': \mathcal{Y} \to \mathcal{J}$, $\mathcal{M}'(\mathcal{M}(\cdot))$ still preserves $\varepsilon$-DP.
\end{lemma}

\begin{lemma}[Basic Composition]
\label{lemma:basic_composition_DP}
    If $\mathcal{M}$ is a composition of DP mechanisms $\mathcal{M}_1, \mathcal{M}_2, \cdots, \mathcal{M}_k$, i.e., $\mathcal{M}(D) = (\mathcal{M}_1(D), \mathcal{M}_2(D), \cdots, \mathcal{M}_k(D))$, where $\mathcal{M}_i$ preserves $\varepsilon$-DP, then $\mathcal{M}$ preserves $k\varepsilon$-DP. 
\end{lemma}

\begin{lemma}[Parallel Composition \citep{Parallel_Composition}]
    Suppose there is a pairwise disjoint splitting $\mathbb{R}_1 \cup \mathbb{R}_2 \cup \cdots \cup \mathbb{R}_k \subseteq \mathbb{R}$, and $\mathcal{M}_1, \mathcal{M}_2, \cdots, \mathcal{M}_k$ is a list of $\varepsilon$-DP mechanisms, then $\mathcal{M}(D) := (\mathcal{M}_1(D \cap \mathbb{R}_1), \mathcal{M}_2(D \cap \mathbb{R}_2), \cdots, \mathcal{M}_k(D \cap \mathbb{R}_k))$ also preserves $\varepsilon$-DP.
\end{lemma}

\subsection{Commonly used DP mechanisms}

Here, we introduce some commonly-used DP mechanisms. Note that all DP mechanisms introduced in the preliminary section work for both bounded and unbounded DP settings by integrating different neighboring relations in the definition.

\subsubsection{Laplace Mechanism}

The most commonly used mechanism is the \textit{Laplace mechanism}. Here, given any query $Q: \mathcal{D} \to \mathbb{R}$, the global sensitivity of $Q$ is defined as $\mathrm{GS}_Q = \max_{D \sim D'} |Q(D) - Q(D')|$. The Laplace mechanism injects Laplace noise proportional to $\mathrm{GS}_Q$ into a query output to preserve DP:
\begin{lemma}[Laplace Mechanism]
    Let $Q$ be a query with global sensitivity $\Delta Q$. The Laplace mechanism
    \begin{equation*}
        \mathcal{M}(D) = Q(D) + \mathrm{Lap}\left(\frac{\Delta Q}{\varepsilon}\right)
    \end{equation*}
    preserves $\varepsilon$-DP, where $\mathrm{Lap}\left(\frac{\Delta Q}{\varepsilon}\right)$ denotes a random variable drawn from the Laplace distribution with scale parameter $\frac{\Delta Q}{\varepsilon}$.
\end{lemma}

The following Laplace tail bound implies that the error of the Laplace mechanism is bounded by $\Delta Q/\varepsilon \ln(1/\beta)$ with probability at least $1 - \beta$:
\begin{lemma}[Laplace Tail Bound]
    \label{lemma:Laplace_tail_bound}
    Let $X$ be a random variable drawn from the unit Laplace distribution $\mathrm{Lap}(1)$. Then, for any $0 < \beta < 1$, the following inequality holds:
    \begin{equation*}
        \Pr(|X| \geq \ln(1/\beta)) \leq \beta.
    \end{equation*}
\end{lemma}

\subsubsection{Clipped Mean estimation}

For mean estimation on the whole real number domain $\mathbb{R}$, we cannot directly utilize the Laplace mechanism to inject noise into the estimated mean, since its sensitivity is unbounded on $\mathbb{R}$. To bound its sensitivity, one commonly used operation is \textit{clipping mechanism}, where we first clip the element into a given range:  
\begin{definition}[Clipping]
    Given a range $[l, r]$, define
    \begin{align*}
        \texttt{Clipped}(x, [l, r]) =
        \begin{cases}
            l, & \text{if } x < l; \\
            x, & \text{if } l \leq x \leq r; \\
            r, & \text{if } x > r.
        \end{cases},
    \end{align*} and 
    \begin{align*}
        \texttt{Clipped}(D, [l, r]) = \{ \texttt{Clipped}(x, [l, r])\}_{x \in D}.
    \end{align*}
\end{definition}
After the clipping, the sensitivity will be bounded by $(l-r)w$ ($w$ is the weight used in mean estimation) and the Laplace mechanism can be adopted.

\subsubsection{The Sparse Vector Technique}
In the \textit{sparse vector technique} (SVT), we consider a (possibly infinite) sequence of queries $Q_1, Q_2, \cdots$, each with a bounded global sensitivity of $1$, and we aim to identify the first query whose output exceeds a given threshold. The procedures of SVT are detailed in Algorithm \ref{alg:SVT}.

\begin{minipage}{0.93\linewidth}
    \centering
    \begin{algorithm}[H]
        \caption{SVT~\citep{SVT}}
        \label{alg:SVT}
        \KwIn{$D$, $\varepsilon$, $\theta$, $Q_1, Q_2, \cdots$}
        \KwOut{i}
        $\tilde{\theta} \gets \theta + \text{Lap}(\frac{2}{\varepsilon})$
        
        \For{$i \gets 1, 2, \cdots$}{
            $\mathcal{M}(\mathcal{P}_1^k \mathcal{P}_2^{n-k})_i(D) \gets Q_i(D) + \text{Lap}(\frac{4}{\varepsilon})$
            
            \If{$\mathcal{M}(\mathcal{P}_1^k \mathcal{P}_2^{n-k})_i(D) > \tilde{\theta}$}{
                Break
            }
        }
        \Return{$i$}
    \end{algorithm}
\end{minipage}

SVT possesses a desirable property: it will neither terminate prematurely nor delay unnecessarily, which are stated in Lemma~\ref{lemma:SVT}.

\begin{lemma}
    Given $D$, $\theta$, $\varepsilon$ and $\beta$, for $Q_1, Q_2, \cdots$, if there exist $k_1$ and $k_2$ such that, $Q_i(D) \leq \theta - \frac{8}{\varepsilon}\log(\frac{4k_1}{\beta})$ holds for all $i \leq k_1$ and $Q_{k_2}(D) \geq \theta + \frac{6}{\varepsilon}\log(\frac{2}{\beta})$, then \emph{SVT} returns an index $i$ such that $i \geq k_1 + 1$ and $Q_i(D) \geq \theta - \frac{6}{\varepsilon}\log(\frac{2k_2}{\beta})$ with probability at least $1 - \beta$.
    \label{lemma:SVT}
\end{lemma}

\subsubsection{The Inverse Sensitivity Mechanism}

For a query $Q$ with discrete output range $\mathcal{Y}$ and $D$, the \textit{inverse sensitivity mechanism} returns $y \in \mathcal{Y}$ such that there exists $D'$ not far from $D$ and $Q(D') = y$. Specifically, first define $d(D, D')$ as the number of differing elements between $D$ and $D'$. Given $Q: \mathcal{D} \to \mathcal{Y}$ and $D$, for any $y \in \mathcal{Y}$, INV assigns a score to denote the distance between $D$ and the preimage set of $y$: $\text{score}_Q(D, y) = -\inf_{D', Q(D')=y} d(D, D')$. Then, INV returns each $y \in \mathcal{Y}$ with the probability
\begin{align*}
    \Pr(\text{INV}(D, Q) = y) = \frac{\exp(\varepsilon\cdot\text{score}_Q(D, y)/2)}{\sum_{z \in \mathcal{Y}}\exp(\varepsilon\cdot\text{score}_Q(D, z)/2)}.
\end{align*}

With INV, we can estimate the quantile $X_{\tau}$ of a dataset $D$, $1 \leq \tau \leq n(D)$, since $-\text{score}_Q(D, y)$ is exactly the number of elements between $X_{\tau}$ and $y$ if we let $Q(D)$ to be query of the $\tau$-quantile of $D$. However, one problem of INV is that, when $\tau$ gets near to $1$ or $n(D)$, it can return some output arbitrarily bad. Therefore, we use INV to find the median, which is an interior point of the dataset.

The details are shown in Algorithm \ref{alg:FindQuantile} and the result enjoys a bounded rank error:

\begin{minipage}{0.93\linewidth}
\centering
\begin{algorithm}[H]
    \caption{FindQuantile \citep{Universal}}
    \label{alg:FindQuantile}
    \KwIn{$D$, $\tau$, $\varepsilon$, $\mathcal{X}$, $\beta$}
    \KwOut{$\tau$-quantile of $D$}
    \If{$\tau < \frac{2}{\varepsilon}\log(\frac{\vert \mathcal{X} \vert}{\beta})$}{
        $\tau' \gets \frac{2}{\varepsilon}\log(\frac{\vert \mathcal{X} \vert}{\beta})$
    }
    \ElseIf{$\tau > n - \frac{2}{\varepsilon}\log(\frac{\vert \mathcal{X} \vert}{\beta})$}
    {
        $\tau' \gets n - \frac{2}{\varepsilon}\log(\frac{\vert \mathcal{X} \vert}{\beta})$
    }
    \Else{
        $\tau' \gets \tau$
    }
    Let $Q$ query the $\tau'$-quantile of $D$
    \For{$x \in \mathcal{X}$}{
        $\text{score}_Q(D, x) \gets -\inf_{D', Q(D')=x} d(D, D')$
    }
    \Return $x \varpropto  \frac{\exp(\varepsilon\cdot\text{score}_Q(D, x)/2)}{\sum_{z \in \mathcal{X}}\exp(\varepsilon\cdot\text{score}_Q(D, z)/2)}$

\end{algorithm}
\end{minipage}

\begin{lemma}
    \label{lemma:quantile_guarantee}
    For a finite ordered set $\mathcal{X}$, suppose $D \in \mathcal{X}^n$, if $n > \frac{4}{\varepsilon} \log(\frac{\vert\mathcal{X}\vert}{\beta})$, then Algorithm \ref{alg:FindQuantile} returns $\tilde{X}_{\tau}$ s.t.
    \begin{align*}
        X_{\tau - \frac{4}{\varepsilon} \log(\frac{\vert\mathcal{X}\vert}{\beta})} \leq \tilde{X}_{\tau} \leq X_{\tau + \frac{4}{\varepsilon} \log(\frac{\vert\mathcal{X}\vert}{\beta})}
    \end{align*}
    with probability at least $1 - \beta$.
\end{lemma}

\subsection{Personalized Differential Privacy}
Traditional DP policies assume that the privacy requirements are uniform across individuals. Recently, a novel DP policy was invented to allow heterogeneous privacy demands, which is called \textit{Personalized Differential Privacy} (PDP) \citep{PDP_2023}. In PDP, each element $x_i$ in $D$ has a distinct privacy budget. Similar to DP model, which defines two types of neighboring relations: change-one and add/remove-one, PDP also supports two corresponding policies: the bounded PDP setting and the unbounded PDP setting.

In the bounded PDP setting, privacy budgets remain public because records are only modified, not inserted or deleted, so each record’s associated privacy budget is preserved.
Therefore, in this setting, all privacy budgets can be stored in a public vector $\boldsymbol{\varepsilon} = \{\varepsilon_1, \ldots, \varepsilon_n\}$, and we use $D \sim_i D'$ to denote that $D$ and $D'$ are neighbors that differ in the $i$-th element, $1 \leq i \leq n$. The PDP definition under the bounded setting is as follows:
\begin{definition}[Bounded Personalized Differential Privacy~\citep{PDP_2023}]
\label{def:PDP_bounded}
    For any $\boldsymbol{\varepsilon} = \{\varepsilon_1, \ldots, \varepsilon_n\} \geq \mathbf{0}$, a mechanism $\mathcal{M}: \mathcal{D} \to \mathcal{K}$ satisfies $\boldsymbol{\varepsilon}$-PDP if, for all neighboring datasets $D \sim_i D'$, and for all measurable subsets $K \subseteq \mathcal{K}$, the following holds for all $i \in \{1, \ldots, n\}$:
    \begin{equation*}
        \Pr(\mathcal{M}(D) \in K) \leq e^{\varepsilon_i} \Pr(\mathcal{M}(D') \in K).
    \end{equation*}
\end{definition}

For the unbounded PDP setting, we have an (infinite) user domain $\mathcal{U}$, where each user $u$ corresponds to value $u.x$. In the statistical setting, $u.x$ is a value drawn from the distribution $\mathcal{P}$. Besides, we have a privacy function $\mathcal{E}$ to assign the privacy budget for each $u$, i.e., $u$ has a privacy budget of $\mathcal{E}(u)$. A dataset $D$ contains only a finite number of users from $\mathcal{U}$. We use $D \sim_u D'$ to denote that $D$ and $D'$ are neighbors if one can be obtained by deleting user $u$ from another. Then, the unbounded PDP is defined as follows:

\begin{definition}[Unbounded Personalized Differential Privacy~\citep{Graham}]
\label{def:PDP_unbounded}
    Given a user universe $\mathcal{U}$ and the privacy function $\mathcal{E}: \mathcal{U} \to [\varepsilon_{\min}, \varepsilon_{\max}]$, where $\varepsilon_{\min}$ and $\varepsilon_{\max}$ are known. A mechanism $\mathcal{M}: \mathcal{D} \to \mathcal{K}$ satisfies $\mathcal{E}$-PDP if, for all neighboring datasets $D \sim_u D'$, and for all measurable subsets $K \subseteq \mathcal{K}$, the following holds for all $u \in \mathcal{U}$:
    \begin{equation*}
        \Pr(\mathcal{M}(D) \in K) \leq e^{\mathcal{E}(u)} \Pr(\mathcal{M}(D') \in K).
    \end{equation*}
\end{definition}

With a little abuse of notation, we use $D = \{x_1,x_2,\dots,x_n\}$ to denote all elements contributed by the users in $D$ and $\boldsymbol{\varepsilon}_D$ to denote their privacy budgets. 

Similar to the DP models, the unbounded PDP provides a stronger privacy guarantee than the bounded PDP. Under the bounded PDP setting, the privacy vector is public, and only the content of the data is protected. In contrast, under the unbounded PDP, the privacy information itself, i.e., $\boldsymbol{\varepsilon}_D$ is also private: although the privacy function $\mathcal{E}$ is publicly known, the presence or absence of any individual user in the dataset is considered private, making the privacy budget private as well. In literature, all PDP statistical estimation works adopt the bounded PDP setting \citep{PDP_2023}, while other data analytics also consider the unbounded PDP setting \citep{Graham}.

Similar to DP, PDP also has a set of desired properties that work under both neighboring relations:

\begin{lemma}[Post-Processing]
\label{lemma:post-processing_PDP}
    If a mechanism $\mathcal{M}: \mathcal{D} \to \mathcal{Y}$ preserves $\boldsymbol{\varepsilon}$-PDP, then for any random mechanism $\mathcal{M}': \mathcal{Y} \to \mathcal{J}$, $\mathcal{M}'(\mathcal{M}(\cdot))$ still preserves $\boldsymbol{\varepsilon}$/$\mathcal{E}$-PDP.
\end{lemma}

\begin{lemma}[Basic Composition]
\label{lemma:basic_composition_PDP}
    If $\mathcal{M}$ is a composition of PDP mechanisms $\mathcal{M}_1, \mathcal{M}_2, \cdots, \mathcal{M}_k$, i.e., $\mathcal{M}(D) = (\mathcal{M}_1(D), \mathcal{M}_2(D), \cdots, \mathcal{M}_k(D))$, where $\mathcal{M}_i$ preserves $\boldsymbol{\varepsilon}$/$\mathcal{E}$-PDP, then $\mathcal{M}$ preserves $k\boldsymbol{\varepsilon}$/$k\mathcal{E}$-PDP. 
\end{lemma}

\begin{lemma}[Parallel Composition \citep{Parallel_Composition}]
\label{lemma:parallel_composition_PDP}
    Suppose there is a pairwise disjoint splitting $\mathbb{R}_1 \cup \mathbb{R}_2 \cup \cdots \cup \mathbb{R}_k \subseteq \mathbb{R}$, and $\mathcal{M}_1, \mathcal{M}_2, \cdots, \mathcal{M}_k$ is a list of $\boldsymbol{\varepsilon}/\mathcal{E}$-PDP mechanisms, then $\mathcal{M}(D) := (\mathcal{M}_1(D \cap \mathbb{R}_1), \mathcal{M}_2(D \cap \mathbb{R}_2), \cdots, \mathcal{M}_k(D \cap \mathbb{R}_k))$ also preserves $\boldsymbol{\varepsilon}$/$\mathcal{E}$-PDP.
\end{lemma}

\subsection{Statistical estimation under PDP}
Here, we provide a brief discussion of the work in statistical estimation under PDP. Existing works~\citep{PDP_2023} study mean estimation over a distribution $\mathcal{P}$ with a domain range $[0,1]$. Under DP, this task can be readily addressed by adding the Laplace noise to the mean. Moving to PDP, since each element has a different privacy budget, it is unfair to treat each element equally. Therefore, one natural idea is to assign each element a weight based on its privacy budget and then compute the weighted sum for the mean estimation. Under the bounded PDP setting (Definition \ref{def:PDP_unbounded}), \citep{PDP_2023} proposes such an \textit{affine differentially private mean} (ADPM) mechanism.

Given the privacy vector $\boldsymbol{\varepsilon}$, ADPM computes its saturated version $\tilde{\boldsymbol{\varepsilon}}$, i.e., $\tilde{\boldsymbol{\varepsilon}}$, and normalizes it into a set of weights. Specifically, it finds the smallest index satisfying $\varepsilon_{k+1} \geq \frac{\sum_{i=1}^k \varepsilon_i^2 + 8}{\sum_{i=1}^k\varepsilon_i}$ and defines $\tilde{\boldsymbol{\varepsilon}}$ to equal $\boldsymbol{\varepsilon}$ except for the last $n-k$ entries, which are left to equal $\tilde{\varepsilon}_{k+1} := \frac{\sum_{i=1}^k \varepsilon_i^2 + 8}{\sum_{i=1}^k\varepsilon_i}$. Then, ADPM uses a weighted mean estimator. The details are shown in Algorithm \ref{alg:ADPM} and it can be shown to preserve PDP with Lemma \ref{lemma:privacy_weighted_mean}.

\begin{minipage}{0.93\linewidth}
    \centering
    \begin{algorithm}[H]
        \caption{ADPM}
        \label{alg:ADPM}
        \KwIn{$D, \boldsymbol{\varepsilon}$, $h$}
        \KwOut{$\widehat{\texttt{Mean}}(D)$}
        $\boldsymbol{\varepsilon} \leftarrow \text{SORT}(\boldsymbol{\varepsilon})$; \tcp*{ascending order}
        $\tilde{\boldsymbol{\varepsilon}}_1 \leftarrow \boldsymbol{\varepsilon}_1$\;
        $k \leftarrow 1$\;
        \While{$k < n$}{
            $\tilde{\boldsymbol{\varepsilon}}_{k+1} \leftarrow \min \left\{ \boldsymbol{\varepsilon}_{k+1}, 
            \frac{\sum_{i=1}^{k} \boldsymbol{\varepsilon}_i^2 + 8}{\sum_{i=1}^{k} \boldsymbol{\varepsilon}_i} \right\}$\;
            $k \leftarrow k + 1$\;
        }
        $\boldsymbol{w} \leftarrow \tilde{\boldsymbol{\varepsilon}} / \sum_{i=1}^n \tilde{\boldsymbol{\varepsilon}}_i$\;
        \Return{$\widehat{\emph{\texttt{Mean}}}(D) = \sum_{i=1}^n \boldsymbol{w}_i\boldsymbol{x}_i + \emph{\texttt{Lap}}(\frac{h}{\sum_{i=1}^n \tilde{\boldsymbol{\varepsilon}}_i)}$} 
    \end{algorithm}
\end{minipage}

\begin{lemma}[Weighted Mean Estimator]
    For any dataset $D$ with bounded range $\vert\emph{\texttt{Range}}(D)\vert$, the following mechanism $\emph{\texttt{WeightedMean}}(D, \mathbf{w}, \varepsilon)$ preserves $\varepsilon$-\emph{DP}.
    \begin{align*}
        \emph{\texttt{WeightedMean}}(D, \mathbf{w}, \varepsilon) = \sum_{i=1}^{n(D)} w_i x_i + \emph{Lap}(\frac{\vert\emph{\texttt{Range}}(D)\vert\cdot\max_{i} w_i}{\varepsilon}),
    \end{align*}
    where $x_i \in D, \mathbf{w} \in \mathbb{R}^{n(D)}, \mathbf{w} \geq \mathbf{0}$.
    \label{lemma:privacy_weighted_mean}
\end{lemma}

Besides, it can be shown that the saturated privacy budget vector $\tilde{\boldsymbol{\varepsilon}}$ found by ADPM enjoys a good property, which will be used in our paper:

\begin{lemma}[\citep{PDP_2023}]
    Given $\boldsymbol{\varepsilon} \geq \mathbf{0}$, we have the following, 
     \begin{align}
        \max_{k_1} \frac{1}{\sum_{i=1}^{k_1} \varepsilon_i + \sqrt{n-{k_1}}} \geq \sqrt{\frac{256}{443}} \frac{\sqrt{\sum_{i=1}^{k_2}\varepsilon_i^2 + (n-k_2)\tilde{\varepsilon}_{k_2+1}^2 + 8}}{2(\sum_{i=1}^{k_2}\varepsilon_i + (n-k_2)\tilde{\varepsilon}_{k_2+1})},
     \end{align}
     where $k_2$ is the smallest index such that $\varepsilon_{k_2+1} \geq \frac{\sum_{i=1}^{k_2} \varepsilon_i^2 + 8}{\sum_{i=1}^k\varepsilon_i}$, and $\tilde{\varepsilon}_{k_2+1} := \frac{\sum_{i=1}^{k_2} \varepsilon_i^2 + 8}{\sum_{i=1}^{k_2}\varepsilon_i}$ by \emph{Algorithm} \ref{alg:ADPM}.
     \label{lemma:pdp_2023_optimality}
\end{lemma}

\subsection{Probabilistic Tools}
Some useful probabilistic tools are presented below, which will be used in the following section.

\begin{theorem}[Relating DP with $\alpha$-divergence \citep{BartheOlmedo2013}]
    A mechanism $\mathcal{M}$ is $\varepsilon$-DP if and only if
    \begin{align*}
        \sup_{D \sim D'} D_{e^\varepsilon} (\mathcal{M}(D) \vert \mathcal{M}(D')) = 0.
    \end{align*}
    The $\alpha$-divergence $D_{\alpha}(\cdot \Vert \cdot)$ $(\alpha \geq 1)$ between two distributions $\mathcal{V}, \mathcal{V}'$ is 
    \begin{align*}
        D_\alpha(\mathcal{V} \vert \mathcal{V}') = \sup_E \left( \Pr_{\mathcal{V}}(E) - \alpha \Pr_{\mathcal{V}'}(E) \right) = \int_z \left[ \frac{d\mathcal{V}}{d\mathcal{V}'}(z) - \alpha \right]_+ d\mathcal{V}'(z) = \sum_{z \in \mathcal{Z}} [\Pr_{\mathcal{V}}(z) - \alpha \Pr_{\mathcal{V}'}(z)]_+,
    \end{align*}
    where $E$ ranges over all measurable subsets of the sample space $\mathcal{Z}$, $[\cdot]_+ = \max\{\cdot, 0\}$, and $\Pr_{\mathcal{V}}, \Pr_{\mathcal{V}'}$ are probabilistic measures induced by $\mathcal{V}, \mathcal{V}'$, respectively. \footnote{Here, $\frac{d\mathcal{V}}{d\mathcal{V}'}$ denotes the Radon-Nikodym derivative between $\mathcal{V}$ and $\mathcal{V}'$ \citep{RN_derivative}}
    \label{theorem:divergence_DP}
\end{theorem}

\begin{definition}[Coupling \citep{Coupling}]
    Given two distributions $\mathcal{V}, \mathcal{V}' \in \mathbb{P}(Z)$, a coupling is a joint distribution $\mathcal{W} \in \mathbb{P}(Z \times Z)$ whose marginals are $\mathcal{V}, \mathcal{V}'$, respectively, i.e., $\mathcal{V} = \sum_{y}\mathcal{W}(y, y'), \mathcal{V}' = \sum_{y'}\mathcal{W}(y, y')$. 
\end{definition}

Coupling always exists and is not unique. Given two distributions $\mathcal{V}, \mathcal{V}'$, among various couplings, the maximal coupling explicitly realizes the total variation (TV) distance between $\mathcal{V}$ and $\mathcal{V}'$:
\begin{align}
    \label{eq:eta_def}
   \eta := \mathrm{TV}(\mathcal{V},\mathcal{V}')= \frac{1}{2}\sum_{y}\vert\mathcal{V}(y)-\mathcal{V}'(y)\vert.
\end{align}
The joint distribution of maximal coupling maximizes the probability $\Pr(\mathcal{V}  = \mathcal{V}')$, which speaks for its name. Let $\mathcal{V}_{\min}(y) = \min\{\mathcal{V}(y), \mathcal{V}'(y)\}$, the maximal coupling is constructed as $\mathcal{W} = (1-\eta)\mathcal{W}_0 + \eta \mathcal{W}_1$, where $\mathcal{W}_0(y, y') = \frac{\mathcal{V}_{\min}(y)1(y = y')}{1 - \eta}$, and $\mathcal{W}_1(y, y') = \frac{(\mathcal{V}(y) - \mathcal{V}_{\min}(y))(\mathcal{V}'(y') - \mathcal{V}_{\min}(y'))}{\eta^2}$. The decomposition of $\mathcal{V}$ and $\mathcal{V}'$ can be obtained by projecting $\mathcal{W}$ along the marginals:
    \begin{equation}
    \label{eq:coupling_induced_decomposition}
       \mathcal{V}=(1-\eta)\mathcal{V}_0+\eta\mathcal{V}_1, \quad \mathcal{V}'=(1-\eta)\mathcal{V}_0+\eta\mathcal{V}_1'.
    \end{equation}
    Here $\mathcal{V}_0$ is the ``agreeing part" and $\mathcal{V}_1,\mathcal{V}_1'$ are the ``disagreeing parts".

\begin{lemma}[Advanced Joint Convexity of $D_\alpha$ \citep{Coupling}]
    Let $\mathcal{V}, \mathcal{V}' \in \mathbb{P}(Z)$ be measures satisfying $\mathcal{V} = (1 - \eta)\mathcal{V}_0 + \eta\mathcal{V}_1$ and $\mathcal{V}' = (1 - \eta)\mathcal{V}_0 + \eta\mathcal{V}_1'$ for some $\eta$, $\mathcal{V}_0$, $\mathcal{V}_1$, and $\mathcal{V}_1'$. Given $\alpha \geq 1$, let $\alpha' = 1 + \eta(\alpha - 1)$ and $\beta = \alpha' / \alpha$. Then the following holds:
    \begin{align*}
        D_{\alpha'}(\mathcal{V} \vert \mathcal{V}') = \eta D_\alpha(\mathcal{V}_1 \vert (1 - \beta)\mathcal{V}_0 + \beta \mathcal{V}_1').
    \end{align*}
    \label{lemma:advanced_joint_convexity}
\end{lemma}

We need an extra notation, i.e., the Markov kernel $M$, for the next lemma: If the mechanism $\mathcal{M}$ is applied upon an input with randomness, it can be associated with a Markov kernel $M$ that operates on a distribution due to the memoryless property \citep{Markov_Kernel}.
\begin{lemma}[Bounding $\alpha$-divergence with Group Privacy \citep{Coupling}]
    Let $M$ denote the Markov kernel associated with $\mathcal{M}(\cdot)$ and $\pi$ denote a coupling between $\mathcal{V}$ and $\mathcal{V}'$, the following holds: 
    \begin{align*}
        D_{e^\varepsilon} (\mathcal{V}M \vert \mathcal{V}'M) \leq \sum_{D, D'} \pi_{D, D'} D_{e^\varepsilon} (\mathcal{M}(D) \vert \mathcal{M}(D')) \leq \sum_{D, D'} \pi_{D, D'} \delta_{\mathcal{M}, d(D, D')}(\varepsilon), 
    \end{align*}
    where $d(D, D')$ is the number of elements that differ in $D$ and $D'$  (hamming distance), and $\delta_{\mathcal{M}, d(D, D')}(\varepsilon)$ is the group privacy profile, which is defined as $\delta_{\mathcal{M}, k}(\varepsilon) = \sup_{d(D, D') \leq k} D_{e^\varepsilon} (\mathcal{M}(D) \vert \mathcal{M}(D'))$.
    \label{lemma:bounding_divergence}
\end{lemma}

\subsection{Concentration Inequalities}

Here, some concentration inequalities that will be used in this work are listed:
\begin{theorem}[Chernoff \citep{mcdiarmid1989method}]
    Given $n$ independent Bernoulli random variables $X_1, X_2, \dots, X_n \sim \emph{\texttt{Ber}}(p_i)$, for any $0 < \eta < 1$, the tail bound of the sum $S_n = \sum_{i=1}^n X_i$ is given by:
    \begin{align*}
        P(S_n \geq (1+\eta) \sum_{i=1}^n p_i) \leq \exp\left(-\frac{\eta^2 \sum_{i=1}^n p_i}{2 + \delta}\right),
    \end{align*}
    and
    \begin{align*}
        P(S_n \leq (1-\eta) \sum_{i=1}^n p_i) \leq \exp\left(-\frac{\eta^2 \sum_{i=1}^n p_i}{2}\right).
    \end{align*}
    \label{theorem:Chernoff}
\end{theorem}

\begin{theorem}[Bernstein \citep{Concentration_Inequalities}]
    Let $X_1, X_2, \dots, X_n$ be independent and identically distributed (i.i.d.) random variables of distribution $\mathcal{P}$ with mean $\mu$ and variance $\sigma^2$, satisfying $|X_i - \mu| \leq M$ almost surely, where $\mu$ is the mean of these variables. Define the sum of deviations from the mean as $S_n = \sum_{i=1}^{n} (X_i - \mu)$. Then, Bernstein's inequality states that for any $t > 0$,
    \begin{align*}
        \Pr(S_n \geq t) &\leq \exp\left( -\frac{t^2}{2(n\sigma^2(X_i) + M t /3)} \right).
    \end{align*}
    \label{theorem:Bernstein}
\end{theorem}

\begin{theorem}[Bernstein, weighted \citep{maurer2009empirical}]
    Let $w_1, w_2, \cdots, w_n$ be fixed weight and let $X_1, X_2, \cdots, X_n$ be i.i.d. random variables of distribution $\mathcal{P}$ with mean $\mu$ and variance $\sigma^2$, and satisfy $|X_i - \mu| \leq M$ almost surely. 
    Then, for any $\delta \in (0,1)$, with probability at least $1 - \delta$, 
    \begin{align*}
        \sum_{i=1}^n w_i (X_i - \mu) \leq \sigma \sqrt{2 \sum_{i=1}^n w_i^2 \ln\left(\frac{1}{\delta}\right)} + \frac{1}{3} M \cdot \max_i w_i \cdot \ln\frac{1}{\delta}.
    \end{align*}
    \label{theorem:weighted_Bernstein}
\end{theorem}

\section{PDP Mean Estimation Under the Bounded PDP setting}
\label{section:upper_bound_bounded}
In this section, we first study mean estimation over Gaussian distributions under the bounded PDP setting. First, in Section~\ref{section:lower_bound_bounded}, we establish a minimax lower bound of this problem. Then we gives our mean estimation framework: Our approach is to first estimate a narrower range of the form $[\mu-\tilde{O}(\sigma),\, \mu+\tilde{O}(\sigma)]$, and then reduce the problem to mean estimation over a bounded domain by clipping all data points into this interval so that we can leverage existing PDP mechanisms for bounded domains \citep{PDP_2023} to perform mean estimation. In Section~\ref{section:diffusion}, we introduce \emph{diffusion}, a fundamental primitive for achieving personalized privacy protection under the bounded PDP setting. With the help of this primitive, in Section~\ref{section:range_estimation_bounded} we propose a PDP-compliant protocol for range estimation over a Gaussian distribution. By clipping all data points into the estimated range and subsequently invoking existing PDP mean estimation protocols, we obtain our final algorithm and the corresponding error bounds are provided in Section~\ref{section:mean_estimation_bounded}, which is optimal up to logarithmic factors.

\subsection{Lower bound}
\label{section:lower_bound_bounded}
First, we aim to establish the lower bound of this problem. For any given $\boldsymbol{\varepsilon}$, we derive a minimax lower bound for the family of Gaussian distributions parameterized by $\sigma$.

\begin{lemma}[\citep{PDP_2023}]
    \label{lemma:PDP_2023_TV}
    For any $\mathcal{M}$ satisfies $\boldsymbol{\varepsilon}$-PDP and any distribution $\mathcal{P}_1, \mathcal{P}_2$, the following inequality holds for any $k \leq n$,
    \begin{align*}
        \Vert \mathcal{M}(\mathcal{P}_1^n)) - \mathcal{M}(\mathcal{P}_2^n)) \Vert_{\mathrm{TV}} \leq 2 \Vert \mathcal{P}_1 - \mathcal{P}_2 \Vert_{\mathrm{TV}} \sum_{i=1}^k (1 - e^{-\varepsilon_i}) + \sqrt{\frac{n-k}{2} D_{\mathrm{KL}}(\mathcal{P}_1 \Vert \mathcal{P}_2)},
    \end{align*}
    where $D_{\mathrm{KL}}(\mathcal{P}_1 \Vert \mathcal{P}_2)$ denotes the Kullback–Leibler divergence defined as $D_{\mathrm{KL}}(\mathcal{P}_1 \Vert \mathcal{P}_2) = \int_{-\infty}^{\infty} \mathcal{P}_1(x) \log(\frac{\mathcal{P}_1(x)}{\mathcal{P}_2(x)})dx$.
\end{lemma}

\begin{theorem}
    Given $\boldsymbol{\varepsilon}$ and $\mathfrak{N}_{\sigma}$ as the family of Gaussian distribution with variance of $\sigma^2$, for any $D \sim \mathcal{N}(\mu, \sigma^2)^n$ and any $\boldsymbol{\varepsilon}$-PDP mechanism $\mathcal{M}$, we have       
        
    \begin{align*}
        \inf_{\mathcal{M}} \sup_{\mathcal{N}(\mu, \sigma^2) \in \mathfrak{N}_{\sigma}} 
        \Pr\left( \vert\mathcal{M}(D) - \mu\vert \ge \max_{k=1, 2, \hdots, n} \frac{1}{\sqrt{2}} \frac{\sigma}{\sum_{i=1}^{k}\varepsilon_i + 2 \sqrt{n - k}}
        \right) \geq \frac{1}{4}
    \end{align*}
    \label{theorem:lower_bound}
\end{theorem}

    \begin{proof}
        First we define $P_{\tau}$ as
        \begin{align*}
            P_{\tau} 
            = \inf_{\mathcal{M}} \sup_{\mathcal{N}(\mu, \sigma^2) \in P} 
            \Pr\left( 
                \vert \mathcal{M}(D) - \mu \vert \ge \tau
            \right).
        \end{align*}
        
        The lower bound can be obtained through the following steps:
        \begin{itemize}
            \item \textbf{Discretization/$2\delta$-Packing} Fix an interval $\delta > 0$, and find a large set of distributions $\{ \mathcal{N}(\mu_i, \sigma^2) \}_{i=1}^M$, such that
            \begin{align*}
                \vert \mu_i, \mu_j \vert \geq 2\delta, \quad i \neq j.
            \end{align*}

            \item \textbf{Relate Estimation to Testing} Let $\{\mathcal{N}_i\}_{i=1}^M := \{ \mathcal{N}(\mu_i, \sigma^2) \}_{i=1}^M$ be a discretization/$2\delta$-packing and $D \sim \mathcal{N}(\mu_j, \sigma^2)$, denote $\psi(\cdot)$ as a tester that tributes the parameter to the nearest packing point, we have
            \begin{align*}
                P_{\tau} \geq \inf_{\mathcal{M}} \text{Pr}\left( \psi(\mathcal{M}(D)) \neq j ~\vert~ D \sim \mathcal{N}(\mu_j, \sigma^2) \right).
            \end{align*}

            \item \textbf{Apply Le Cam's Method} For a binary test, Le Cam's method \citep{Duchi02012018} shows that
            \begin{align*}
                \inf_{\psi}\text{Pr}[\psi(\mathcal{M}(\cdot))~\text{makes an error}] = \frac{1}{2} - \frac{1}{2} \Vert \mathcal{M}(\mathcal{N}_1^n) - \mathcal{M}(\mathcal{N}_2^n) \Vert_{\text{TV}}.
            \end{align*}
        \end{itemize}
    
    Combining the above steps yields the minimax lower bound after selecting $\mathcal{N}_1$ and $\mathcal{N}_2$,
    \begin{align}
        P_{\tau} \geq \frac{1}{2} - \frac{1}{2} \Vert \mathcal{M}(\mathcal{N}_1^n) - \mathcal{M}(\mathcal{N}_2^n) \Vert_{\text{TV}}
        \label{eq:minimax_lower_bound_P_1_P_2}
    \end{align}
    
    Now, we use Lemma \ref{lemma:PDP_2023_TV} to relate the output distance to the input distance. Before that, we first do some preparation.
    
    Take $\mathcal{N}_1$ and $\mathcal{N}_2$ to be two Gaussian distributions $\mathcal{N}(-\tau, \sigma^2), \mathcal{N}(\tau, \sigma^2)$, first we check $\Vert \mathcal{N}_1 - \mathcal{N}_2 \Vert_{\mathrm{TV}}$:
    \begin{equation*}
        \begin{aligned}
            \Vert \mathcal{N}_1 - \mathcal{N}_2 \Vert_{\mathrm{TV}} & = \frac{1}{2} \int_{-\infty}^{\infty} \vert p_1(x) - p_2(x) \vert dx \\ 
            & = \int_{0}^{\infty} p_1(x) - p_2(x) dx \\
            & = \int_{0}^{2\tau} p_1(x) dx = \frac{1}{\sqrt{2}} \texttt{erf}(\frac{\tau}{\sigma}),
        \end{aligned}
    \end{equation*}
    where $\texttt{erf}(x) := \frac{1}{\sqrt{\pi}} \int_{-x}^{x} e^{-t^2} dt$.


    For $D_{\mathrm{KL}}(\mathcal{N}_1 \Vert \mathcal{N}_2)$, we have
    \begin{align*}
        D_{\mathrm{KL}}(\mathcal{N}_1 \Vert \mathcal{N}_2) = 2\frac{\tau^2}{\sigma^2}.
    \end{align*}

    Therefore, by Lemma \ref{lemma:PDP_2023_TV}, we have
    \begin{align*}
        \Vert \mathcal{M}(\mathcal{N}_1^n) - \mathcal{M}(\mathcal{N}_2^n) \Vert_{\mathrm{TV}} \leq \frac{1}{\sqrt{2}} \texttt{erf}(\frac{\tau}{\sigma}) \sum_{i=1}^{k}(1 - e^{-\varepsilon_i}) + \frac{\tau}{\sigma}\sqrt{2(n-k)},
    \end{align*}
    and
    \begin{equation*}
        \begin{aligned}
            P_{\tau} & \geq \frac{1}{2} (1 - \frac{1}{\sqrt{2}} \texttt{erf}(\frac{\tau}{\sigma}) \sum_{i=1}^{k} (1 - e^{-\varepsilon_i}) - \frac{\tau}{\sigma}\sqrt{2(n-k)}) \\
            & \geq \frac{1}{2} - \frac{\tau}{2\sigma}(\frac{1}{\sqrt{2}} \sum_{i=1}^{k} \varepsilon_i + \sqrt{2(n-k)}).
        \end{aligned}
    \end{equation*}

    Let $\tau = \frac{1}{\sqrt{2}} \frac{\sigma}{\sum_{i=1}^{k}\varepsilon_i + 2 \sqrt{n - k}}$, we have $P_{\tau} \geq \frac{1}{4}$. Therefore, the minimax lower bound on the Gaussian distribution family is given as: 

    \begin{align*}
        \inf_{\mathcal{M}} \sup_{\mathcal{N}(\mu, \sigma^2) \in \mathfrak{N}_{\sigma}} 
        \Pr\left( \vert\mathcal{M}(D) - \mu\vert \ge \frac{1}{\sqrt{2}} \frac{\sigma}{\sum_{i=1}^{k}\varepsilon_i + 2 \sqrt{n - k}}
        \right) \geq \frac{1}{4},
    \end{align*}
    for all $k \in 1, 2, \hdots, n$. Take the maximum yields the lower bound.
\end{proof}


\subsection{Diffusion: A Basic Primitive for Bounded PDP}
\label{section:diffusion}
To achieve PDP, besides introducing different weights in the estimated mean, another natural approach is through \textit{sampling}. The high-level idea is to assign different sampling rates to individual elements in the dataset. Since sampling can amplify privacy protection proportionally to the sampling rate, this strategy effectively accommodates the heterogeneous privacy requirements of PDP.

However, this idea cannot be directly applied under the bounded PDP setting. More precisely, sampling-based amplification holds under both the bounded and the unbounded PDP settings. Yet, to the best of our knowledge, no existing work has investigated the bounded PDP setting with nonuniform sampling. Fundamentally, to achieve sampling amplification, given two neighboring datasets $D$ and $D'$ and their sampled distributions $\mathcal{V}$ and $\mathcal{V}'$, we need to construct a coupling $\mathcal{W}(S,S')$ such that, for every support, $S$ and $S'$ are neighbors. Under the bounded PDP setting, constructing such a coupling is straightforward for uniform sampling. In contrast, with nonuniform sampling, ensuring that $S$ and $S'$ remain change-one neighbors—that is, they share the same size—renders the sampling events unequally correlated, making it difficult to establish a valid coupling. In the unbounded PDP setting, however, the problem becomes much simpler: Poisson sampling allows each element to be sampled independently with distinct rates, enabling an easy construction of the desired coupling.

Therefore, to achieve a similar privacy amplification as sampling under the bounded PDP setting, we propose a new technique, i.e., \textit{diffusion}. This will serve as a fundamental building block of our mean estimation algorithm. The key idea is simple: We use some postprocessing to transfer neighbors in the insert/delete-one setting into change-one neighbors. More precisely, first, we apply Poisson sampling process. To ensure the sampled datasets are change-one neighbors, instead of removing the unsampled elements from the dataset, we replace them with a placeholder value $\perp$, which represents a pre-assigned constant or marker that does not interfere with the internal computation of the mechanism. For example, in sum estimation, $\perp$ can be set to $0$, as it does not affect the result. In other cases, we apply additional processing to $\perp$ to ensure it does not interfere with the computation. The diffusion mechanism is defined as follows:

\begin{definition}[Diffusion]
    Given any dataset $D$, any placeholder $\perp$, and any set of probabilities $\{p_i\}_{i=1}^n$, $\mathcal{F}(D, \perp, \{p_i\}_{i=1}^n)$ is called a diffusion on $D$, if it outputs a dataset $D_{\mathcal{F}}$ such that for every $x_i \in D$,
    \begin{align*}
        x_i \in D_{\mathcal{F}}~\text{, or}~\perp \in D_{\mathcal{F}},
    \end{align*}
    with probability $p_i$ and $1-p_i$, respectively. That is, the diffusion keeps $x_i$ with probability $p_i$, or replace it with $\perp$, with probability $1-p_i$.
\end{definition}

Analogous to the sampling operator in DP, the diffusion operation can bring a privacy amplification effect in PDP:
\begin{theorem}[Diffusing Amplification in PDP]
    Given $\mathcal{F}(D, \perp, \{p_i\}_{i=1}^n)$ as a diffuser on $D$, with diffusing parameters $\perp$ and $\{p_j\}_{j=1}^n$.  If a mechanism $\mathcal{M}$ satisfies $\varepsilon$-DP, then the composite mechanism $\mathcal{M}^{\mathcal{F}} := \mathcal{M}(\mathcal{F}(\cdot))$ satisfies $\boldsymbol{\varepsilon}$-PDP, both under the bounded PDP setting, where $e^{\varepsilon_i} = 1 + p_i(e^{\varepsilon} - 1)$.
    \label{theorem:diffusing_amplification}
\end{theorem}

\begin{proof}
    For any $D\sim_i D'$,  our target is to show 
    \begin{equation}
        \label{eq:target_divergence}
        D_{e^{\varepsilon_i}}(\mathcal{M}^{\mathcal{F}}(D) \vert \mathcal{M}^{\mathcal{F}}(D')) = 0.
    \end{equation}

    To show \eqref{eq:target_divergence}, we first use Lemma \ref{lemma:advanced_joint_convexity} to decompose the $\alpha$-divergence between $\mathcal{M}^{\mathcal{F}}(D)$ and $\mathcal{M}^{\mathcal{F}}(D')$. Let $\mathcal{V}_1, \mathcal{V}_0, \mathcal{V}_1'$ be decompositions of $\mathcal{F}(D, \perp, \{p_i\}_{i=1}^n)$ and $\mathcal{F}(D', \perp, \{p_i\}_{i=1}^n)$ induced by the maximal coupling as \eqref{eq:coupling_induced_decomposition}, and $M$ denotes the Markov kernel associated with mechanism $\mathcal{M}(\cdot)$, which is notated for operations upon $\mathcal{V}_1, \mathcal{V}_0$ and $\mathcal{V}_1'$. Plug $\alpha = e^{\varepsilon}$ into Lemma \ref{lemma:advanced_joint_convexity} and by the convexity of $D_{\alpha}(\mathcal{Y}_1 \Vert \mathcal{Y}_1)$ in $\mathcal{Y}_2$ \citep{Coupling}:
    \begin{equation}
        \label{eq:divergence_decomposition}
        D_{e^{\varepsilon_i}}(\mathcal{M}^{\mathcal{F}}(D) \vert \mathcal{M}^{\mathcal{F}}(D')) \leq \eta \cdot ( (1 - \beta) D_{e^{\varepsilon}}(\mathcal{V}_1M \vert \mathcal{V}_0M) + \beta D_{e^{\varepsilon}}(\mathcal{V}_1M \vert \mathcal{V}_1'M)),
    \end{equation}
    where $e^{\varepsilon_i} = 1 + \eta (e^{\varepsilon} - 1)$, also by Lemma \ref{lemma:advanced_joint_convexity}, and $\eta = \texttt{TV}(\mathcal{F}(D, \perp, \{p_i\}_{i=1}^n), \mathcal{F}(D', \perp, \{p_i\}_{i=1}^n))$. One can note that the TV distance between $\mathcal{F}(D, \perp, \{p_i\}_{i=1}^n)$ and $\mathcal{F}(D', \perp, \{p_i\}_{i=1}^n)$ can be expressed as:
    \begin{align*}
        \eta = \frac{1}{2}~\sum_{\mathbf{y} \ni x_i \text{or} x_i'}\vert\Pr(\mathcal{F}(D, \perp, \{p_i\}_{i=1}^n)=\mathbf{y}) - \Pr(\mathcal{F}(D', \perp, \{p_i\}_{i=1}^n)=\mathbf{y})\vert = \frac{p_i + p_i}{2} = p_i,
    \end{align*}
    which equals the diffusing rate of the $i$-th element.

    Now it remains to bound $D_{e^{\varepsilon}}(\mathcal{V}_1 M \Vert \mathcal{V}_0 M)$ and $D_{e^{\varepsilon}}(\mathcal{V}_1 M \Vert \mathcal{V}_1' M)$. Given $\mathcal{V} = (1-\eta)\mathcal{V}_0 + \eta\mathcal{V}_1$ and $\mathcal{V}' = (1-\eta)\mathcal{V}_0 + \eta\mathcal{V}_1'$, the couplings between $\mathcal{V}_1'$ and $\mathcal{V}_1$, and between $\mathcal{V}_1'$ and $\mathcal{V}_0$ are as follows: Obtain $\mathbf{z}$ as $\mathbf{z} = \mathcal{F}(D - \{x_i\}, \perp, \{p_j\}_{j=1}^n - \{p_i\})$,
    \begin{itemize}
        \item Since $\mathcal{V}_1'$ and $\mathcal{V}_1$ stand for the distributions conditional on events that $x_i'$ is not diffused, and $x_i$ is not diffused, respectively, we have $\mathcal{V}_1'(\mathbf{z} \cup x_i') = \mathcal{V}_1(\mathbf{z} \cup x_i)$,
        \item Since $\mathcal{V}_0$ stands for the distributions conditional on the event that $x_i$ is diffused, we have $\mathcal{V}_0(\mathbf{z} \cup v) = \mathcal{V}_1(\mathbf{z} \cup x_i)$, which follows from the independence of the diffusion of the elements.
    \end{itemize}

    Note that the supports of these two constructed couplings differ by one element while keeping the probabilities aligned, applying Lemma \ref{lemma:bounding_divergence} to the two terms in the right hand side of \eqref{eq:divergence_decomposition} yields:
    \begin{align*}
        D_{e^{\varepsilon_i}}(\mathcal{M}^{\mathcal{F}}(D) \vert \mathcal{M}^{\mathcal{F}}(D)) \leq \eta~\delta_{\mathcal{M}}(\varepsilon)\sum_{\mathbf{y}, \mathbf{y}', d(\mathbf{y}, \mathbf{y}')=1} \pi_{\mathbf{y}, \mathbf{y}'} = 0,
    \end{align*}
    which implies $\mathcal{M}^{\mathcal{F}}(\cdot)$ outputs an $\varepsilon_i$-indistinguishable result by theorem \ref{theorem:divergence_DP} for user $i$.
    
    Therefore, $\mathcal{M}^{\mathcal{F}}(\cdot)$ satisfies $\boldsymbol{\varepsilon}$-PDP, where $e^{\varepsilon_i} = 1 + p_i(e^{\varepsilon} - 1)$.
\end{proof}

Theorem \ref{theorem:diffusing_amplification} can serve as a bridge between DP mechanisms and PDP mechanisms under the bounded PDP setting. By assigning different diffusing rates for elements in the dataset, a DP mechanism can be called after diffusion with a uniform privacy budget upon the sampled dataset, which can achieve PDP. The following mean estimation algorithm is built on this idea.

\subsection{PDP Range Estimation}
\label{section:range_estimation_bounded}
After introducing the diffusion framework above, we now discuss how to perform range estimation for a Gaussian distribution under PDP with the help of diffusion. As mentioned at the beginning of Section~\ref{section:upper_bound_bounded}, our goal is to estimate a range of $[\mu - \tilde{O}(\sigma),\, \mu + \tilde{O}(\sigma)]$. Due to the concentration properties of the Gaussian distribution, estimating such an interval is essentially equivalent to estimating the range of a dataset drawn from this distribution. Therefore, a straightforward approach is to first apply the diffusion framework to enforce personalized privacy protection, and then run a DP range estimation algorithm, such as that in~\citep{Universal}, on the diffused dataset to obtain a PDP-compliant range estimate. However, this seemingly straightforward approach encounters two key challenges:

1) The first challenge is how to set the diffusing rates. Intuitively, to meet the requirements of PDP, the diffusing rate of each element should be proportional to its allocated privacy budget. Moreover, it is crucial to select an appropriate overall scale for these rates. On one hand, smaller rates provide stronger privacy amplification, thereby reducing the DP error; on the other hand, excessively small rates can cause large sampling error. Therefore, the key is to identify diffusing rates that not only satisfy the PDP constraints but also strike a favorable balance between these sources of error.

2) The second challenge is how to accommodate the existing DP range estimator \citep{Universal} over the diffused dataset. So far, the methods in \citep{Universal} are designed over a dataset with a public data size and do not support the placeholders. After diffusion, the data size becomes private information. Therefore, additional processing is required to support the placeholders and the sensitive data size. 

Below, we give our solution to address these two challenges.

\paragraph{Setting Diffusing Ratio}

Let us first discuss how to determine the diffusing rates $\{p_i\}_{i=1}^n$. As mentioned earlier, to ensure the PDP requirement, the diffusing rates must be proportional to the corresponding privacy budgets. Specifically, we introduce a threshold $\tau$ to divide the elements into two groups based on whether the privacy budget is below $\tau$.
For those with privacy budgets smaller than $\tau$, the diffusing rates are set proportional to their privacy budgets, with privacy budget $\tau$ assigned a probability of $1$; for the others, all the rates are simply fixed at $1$. According to Theorem \ref{theorem:diffusing_amplification}, after this diffusion step, a DP mechanism with privacy budget $\tau$ can be applied to the diffused dataset.

Now, the problem reduces to finding an appropriate value of the threshold $\tau$: a too large $\tau$ will filter out too many elements, leading to a large sampling error, while a too small $\tau$ will result in a substantial DP error.
To achieve a good tradeoff between sampling error and DP error, we identify the smallest index $k$ such that 
\begin{align*}
\varepsilon_{k+1} \ge 
\frac{\sum_{i=1}^k \varepsilon_i^2 + 8}{\sum_{i=1}^k \varepsilon_i},
\end{align*}
and then set $\tau := \frac{\sum_{i=1}^k \varepsilon_i^2 + 8}{\sum_{i=1}^k \varepsilon_i}$. Later, we will show using this threshold indeed achieves a good tradeoff between sampling error and DP error, by showing that our final error for PDP mean estimation matches the lower bound derived in Section~\ref{section:lower_bound_bounded}.
Notably, a threshold of the same form has also been used in \citep{PDP_2023}, where $\tau$ is applied to rescale weights in order to achieve personalized privacy protection for mean estimation over a bounded domain, and this setting achieves optimal error as well. It is interesting that, even though \citep{PDP_2023} uses completely different techniques and addresses a different task from our PDP range estimator, the same threshold yields optimal performance in both settings.

Above all, we set $\tau = \frac{\sum_{i=1}^k \varepsilon_i^2 + 8}{\sum_{i=1}^k\varepsilon_i}$.  For each data $x_i$ with privacy budget $\varepsilon_i$ has the diffusing rate 
\begin{align}
    p_i =
    \left\{
    \begin{aligned}
        & \varepsilon_i/2\tau, & i \leq k, \\
        & 1/2, & i > k,
    \end{aligned}
    \right.
    \label{eq:diffusing_rate}
\end{align}
where $k$ is the first index satisfying $\varepsilon_{k+1} \geq \frac{\sum_{i=1}^k \varepsilon_i^2 + 8}{\sum_{i=1}^k\varepsilon_i}$ \citep{PDP_2023}. 

\paragraph{Estimating Range of Diffused dataset}

Now, we estimate the range of the diffused dataset using the privacy budget $\tau$ computed as \eqref{eq:diffusing_rate}. Let the diffused dataset be denoted by 
$D_{\mathcal{F}} = \mathcal{F}(D, \perp, \{p_i\}_{i=1}^n)$; in particular, we set $\perp = 0$. The remaining task is to adapt the DP range estimator of~\citep{Universal} so that it can operate on the diffused dataset. 
In~\citep{Universal}, to estimate the range of a Gaussian distribution, the procedure is to first locate a point close to $\mu$ by estimating the median of the dataset, and then determines an interval of appropriate width centered at $\mu$. Before performing these two steps, it is necessary to choose a suitable discretization granularity for the continuous domain, so that the data can be discretized and an efficient search can be carried out. 
Below, we introduce these three steps briefly and only highlight our modifications over the original algorithm of \citep{Universal} required to support the diffused dataset.

\textbf{1) Discretizing the continuous domain}
To determine an appropriate discretization granularity, we use a lower bound on the standard deviation $\sigma$ of the distribution. To achieve this goal, we first construct \textit{differential dataset} $G$ by randomly pairing elements $(x_i, x_j)$ drawn from $D$ and taking the absolute differences of these pairs. We then approximate the median of the diffused version of $G$, which yields a constant-factor lower bound on $\sigma$. 
More precisely, we find an interval $[0, 2^i]$ with $i\in \mathbb{Z}$ such that $\texttt{Count}(G, 2^i) \approx \frac{\vert G \vert}{2}$. The first instance compares the sequence $\{\texttt{Count}(G, 2^i)\}_{i = 0, 1, \hdots}$ with $\lvert G \rvert / 2$ and identifies the smallest $i$ for which $\texttt{Count}(G, 2^i) > \lvert G \rvert / 2$. The second instance compares the negated sequence $-\{\texttt{Count}(G, 2^{-j})\}_{j = 0, 1, \hdots}$ with $-\lvert G \rvert / 2$ and finds the smallest $j$ such that $\texttt{Count}(G, 2^{-j}) < \lvert G \rvert / 2$.

However, since we set $\perp = 0$ as a placeholder, the diffused dataset may contain much more than half zeros, making the above estimator fail. To address this, we consider only the nonzero elements in $D_{\mathcal{F}}$, i.e., we construct $G$ from $D_{\mathcal{F}}^{+}=D_{\mathcal{F}}\cap \mathbb{R}^{+}.$ But this makes the data size $\lvert G\rvert$ sensitive, making it no longer usable directly as the threshold in SVT since it requires the threshold to be a public parameter.
Notably, making the data size public would forfeit the sampling amplification benefits provided by diffusion.
To resolve this issue, we instead set the threshold to zero and incorporate the sensitive data size into the count query. Specifically, in the first SVT, we evaluate $\texttt{Count}(G, 2^{i}) - \frac{\vert G\vert}{2}, i=0,1,\dots,$
and return the first index for which this value becomes positive. The second SVT is handled symmetrically. The detailed algorithm is shown in Algorithm \ref{alg:FindBucketSize}. The yielded discretization granularity $b$ enjoys the following guarantee:

\begin{minipage}{0.93\linewidth}
    \centering
    \begin{algorithm}[H]
        \caption{Discretization}
        \label{alg:FindBucketSize}
        \KwIn{$D$, $\varepsilon$, $\beta$}
        \KwOut{Discretized dataset $D$, $b$}
        $D \gets D \cap \mathbb{R}^{+}$\;
        Construct $G$ such that $G \ni \vert x_{i} - x_{j}\vert, x_i, x_j \in D$
        $\tilde{i} \gets \text{SVT}\left(0, \frac{\varepsilon}{2}, \texttt{Count}(G, 2^0) - \frac{\vert G \vert}{2}, \texttt{Count}(G, 2^1) - \frac{\vert G \vert}{2}, \dots \right)$\;
        $\tilde{j} \gets \text{SVT}\left(0, \frac{\varepsilon}{2},\frac{\vert G \vert}{2} -\texttt{Count}(G, 2^0), \frac{\vert G \vert}{2}-\texttt{Count}(G, 2^{-1}), \dots \right)$\;
        \If{$\tilde{i} > 1$}{
            $b \gets 2^{\tilde{i} - 2}$\;
        }
        \Else{
            $b \gets 2^{-\tilde{j}}$\;
        }
        $b \gets \frac{b}{2}$\;
        \For{$x \in D$}{
            $x \gets b \cdot \lfloor \frac{x}{b} \rfloor$;            \tcp{Assign $x$ to its quantization bin}
        }
        \Return{$D, b$}
    \end{algorithm}
    \end{minipage}

    \begin{lemma}
        Given $\boldsymbol{\varepsilon}$, $\beta$, for any $D \sim \mathcal{N}(\mu, \sigma^2)^n$, let $D_{\mathcal{F}}$ be a diffused dataset with the diffusion ratio setting as \eqref{eq:diffusing_rate}, if $n = \Omega(\frac{1}{\varepsilon_1}\log(\frac{1}{\beta}) + \frac{1}{\varepsilon_1}\log\log(\frac{1}{\sigma}))$, then running Algorithm \ref{alg:FindBucketSize} over $D_{\mathcal{F}}$ with privacy budget $\tau$ returns $b$ such that
        \begin{align*}
            \frac{1}{8} \sigma \leq b \leq \sigma,
        \end{align*}
        with probability at least $1 - \beta$.
        \label{lemma:bucketsize}
    \end{lemma}

    \begin{proof}
         The proof follows the idea of the proof of Theorem 4.3 in \citep{Universal}. The main difference is that the size of the differential dataset needs extra analysis:
         Since $G$ is constructed by pairing the undiffused elements in $D$, its size $\vert G \vert$ is a Bernoulli sum, and $\mu(\vert G \vert) = \frac{1}{4} \frac{\sum_{i=1}^{k} \varepsilon_i + (n - k)\tau}{\tau} \pm O(1)$ (note that $\varepsilon_i$s are sorted in an ascending manner). By Chernoff's inequality, if $\frac{\sum_{i=1}^{k} \varepsilon_i + (n - k)\tau}{\tau} = \Omega(\log(\frac{1}{\beta}))$, with probability at least $1- \beta/2$, $\vert G \vert \geq \frac{1}{8} \frac{\sum_{i=1}^{k} \varepsilon_i + (n - k)\tau}{\tau}$.

        Now we see the relation between $n$ and $\vert G\vert$, using a case analysis:
        \begin{itemize}
            \item $k \leq \frac{n}{2}:$ In this case, $n - k \geq \frac{n}{2}$, therefore $n = \Omega(\frac{1}{\varepsilon_1}\log(\frac{1}{\beta}) + \frac{1}{\varepsilon_1}\log\log(\frac{1}{\sigma}))$ suffices to satisfy:
            a) $\frac{\sum_{i=1}^{k} \varepsilon_i + (n - k)\tau}{\tau} = \Omega(\log(\frac{1}{\beta}))$, and incurs to satisfy b) $\vert G \vert = \Omega(\frac{1}{\tau}\log(\frac{1}{\beta}) + \frac{1}{\tau}\log\log(\frac{1}{\sigma}))$;

            \item $k > \frac{n}{2}:$ In this case, $\frac{\sum_{i=1}^k\varepsilon_i}{\tau} \geq \frac{n}{2} \frac{\varepsilon_1}{\tau}$, $n = \Omega(\frac{1}{\varepsilon_1}\log(\frac{1}{\beta}) + \frac{1}{\varepsilon_1}\log\log(\frac{1}{\sigma}))$ shows $\vert G \vert = \Omega(\frac{1}{\tau}\log(\frac{1}{\beta}) + \frac{1}{\tau}\log\log(\frac{1}{\sigma}))$.
        \end{itemize}
        
        Then by Theorem 4.3 in \citep{Universal}, if $n = \Omega(\frac{1}{\varepsilon_1}\log(\frac{1}{\beta}) + \frac{1}{\varepsilon_1}\log\log(\frac{1}{\sigma}))$, then with probability at least $1 - \beta$, Algorithm \ref{alg:FindBucketSize} returns a $b$ such that $\frac{1}{8} \sigma \leq b \leq \sigma$ (Computation of the constants $\frac{1}{8}$ and $1$ uses the quantiles of the standard Gaussian distribution).
    \end{proof}
    
    \textbf{2) Median Estimation}
    After discretizing $D_{\mathcal{F}}$, \footnote{After discretizing, all datasets are refered as the discretized version in this subsection.} we proceed to estimate its median. As introduced in the Preliminaries, to estimate the median of a given dataset with INV, we first need to estimate a data range. Such range estimation can be coarse since it only has a logarithmic dependency in the mean estimation error. Following the approach proposed in \citep{Universal}, this coarse range can be obtained by estimating the maximum absolute value of the elements in $D_{\mathcal{F}}$, denoted by $\widehat{\texttt{Radius}}(D_{\mathcal{F}})$, such that the interval $[-\widehat{\texttt{Radius}}(D_{\mathcal{F}}), \widehat{\texttt{Radius}}(D_{\mathcal{F}})]$ covers most of elements in $D_{\mathcal{F}}$.
    This is also done by running SVT, where we compare the sequence $\{\texttt{Count}(D_{\mathcal{F}}, 2^i b)\}_{i\geq0}$ with $\frac{\vert D_{\mathcal{F}} \vert}{2}$ and find the first $i$ such that $\texttt{Count}(D_{\mathcal{F}}, 2^i b) > n(D_{\mathcal{F}}) - \frac{6}{\varepsilon}\log(\frac{2}{\beta})$. Similar to Algorithm \ref{alg:FindBucketSize}, the data size now becomes sensitive information, thus we also need to move it to the query side. Detailed procedures are in Algorithm \ref{alg:EstimateRadius}. Notably, the range obtained is only a coarse estimate, and a finer estimation is needed, particularly in cases where $\lvert \mu \rvert$ is substantially larger than $\sigma$.
    After obtaining this coarse range, we can compute a privatized median $\widehat{\texttt{Median}}(D_{\mathcal{F}}^{+})$ within the bounded interval $[-\widehat{\texttt{Radius}}(D_{\mathcal{F}}), \widehat{\texttt{Radius}}(D_{\mathcal{F}})]$, using the inverse-sensitivity mechanism introduced in the Preliminaries. 
    
    \begin{minipage}{0.93\linewidth}
    \centering
    \begin{algorithm}[H]
        \caption{EstimateRadius \citep{Universal}}
        \label{alg:EstimateRadius}
        \KwIn{$D$, $\varepsilon$, $b$, $\beta$}
        \KwOut{$\widehat{\texttt{Radius}}(D)$}
        $i \gets \text{SVT}\left(-\frac{6}{\varepsilon}\log(\frac{2}{\beta}), \varepsilon, \texttt{Count}(D, 2^0\cdot b) - \vert D \vert, \texttt{Count}(D, 2^1 \cdot b) - \vert D \vert, \dots \right)$\;
        \Return{} $2^{i-2}\cdot b$
    \end{algorithm}
    \end{minipage}

    \textbf{3) Range Estimation}
    With the privatized median $\widehat{\texttt{Median}}(D_{\mathcal{F}})$ in hand, we estimate a region centered at this median that covers most elements of the diffused dataset, which serves as the estimated range. To achieve this goal, we first shift $D_{\mathcal{F}}^{+}$ by subtracting $\widehat{\texttt{Median}}(D_{\mathcal{F}})$, and denote the resulting dataset as $D_{\mathcal{F}}^{c}$. We then estimate its radius, $\widehat{\texttt{Radius}}(D_{\mathcal{F}}^{c})$. Finally, we use $\widehat{\texttt{Range}}(D) = [- \widehat{\texttt{Radius}}(D_{\mathcal{F}}^c) + \widehat{\texttt{Median}}(D_{\mathcal{F}}^{+}), \widehat{\texttt{Radius}}(D_{\mathcal{F}}^c) + \widehat{\texttt{Median}}(D_{\mathcal{F}}^{+})]$ as our estimated range of $D$.


\begin{minipage}{0.93\linewidth}
    \centering
    \begin{algorithm}[H]
        \caption{PDPEstimateRange}
        \label{alg:EstimateRange}
        \KwIn{$D$, $\boldsymbol{\varepsilon}$, $\beta$}
        \KwOut{$\widehat{\texttt{Range}}(D)$}
        $k \gets \min(\arg\max_j \varepsilon_j < \frac{\sum_{i=1}^j \varepsilon_i^2 + 8}{\sum_{i=1}^j\varepsilon_i})$; \tcp{Find the maximum non-saturation index}
        $\tau \gets \frac{\sum_{i=1}^k \varepsilon_i^2 + 8}{\sum_{i=1}^k\varepsilon_i}$\;
        \For{$i \gets 1, 2, \cdots, n$}
        {
        \If{$i \leq k$}
        {$p_i \gets \varepsilon_i/\tau$}
        \Else{
        $p_i \gets 1$
        }
        }
        $D_{\mathcal{F}} = \mathcal{F}(D, 0, \{p_i\}_{i=1}^n)$; \tcp{Diffuse $D$ with diffusing rates $\{p_i\}_{i=1}^n$}
        $D_{\mathcal{F}}, b \gets \texttt{Discretization}(D_{\mathcal{F}}, \frac{\tau}{4}, \{p_i\}_{i=1}^n, \frac{\beta}{4})$
        $D_{\mathcal{F}}^{+} = D_{\mathcal{F}} \cap \mathbb{R}^{+}$\;
        $\widehat{\texttt{Radius}}(D_{\mathcal{F}}) \gets \texttt{EstimateRadius}(D_{\mathcal{F}}, \frac{\tau}{4}, b, \frac{\beta}{4})$\;
        $\widehat{\texttt{Median}}(D_{\mathcal{F}}^{+}) \gets \texttt{FindQuantile}(D_{\mathcal{F}}^{+} \cap [-\widehat{\texttt{Radius}}(D_{\mathcal{F}}), \widehat{\texttt{Radius}}(D_{\mathcal{F}})], \frac{\sum_{i=1}^{k} \varepsilon_i + (n - k)\tau}{2\tau}, \frac{\tau}{4}, [-\widehat{\texttt{Radius}}(D_{\mathcal{F}}), \widehat{\texttt{Radius}}(D_{\mathcal{F}})] \cap b\mathbb{Z}, \frac{\beta}{4})$\;
        $D_{\mathcal{F}}^c \gets D_{\mathcal{F}}^{+} - \widehat{\texttt{Median}}(D_{\mathcal{F}}^{+})$
        $\widehat{\texttt{Radius}}(D_{\mathcal{F}}^c) \gets \texttt{EstimateRadius}(D_{\mathcal{F}}^c, \frac{\tau}{4}, b, \frac{\beta}{4})$\;
        $\widehat{\texttt{Range}}(D) \gets [- \widehat{\texttt{Radius}}(D_{\mathcal{F}}^c) + \widehat{\texttt{Median}}(D_{\mathcal{F}}^{+}), \widehat{\texttt{Radius}}(D_{\mathcal{F}}^c) + \widehat{\texttt{Median}}(D_{\mathcal{F}}^{+})]$\;
        
        \Return{$\widehat{\emph{\texttt{Range}}}(D)$}
    \end{algorithm}
    \end{minipage}

Above all, our PDP range estimator is detailed in Algorithm \ref{alg:EstimateRange}. We can show that it yields an estimation of range that covers most elements in $D$ with high probability:
\begin{lemma}
    \label{lemma:range_estimation}
    Given $\boldsymbol{\varepsilon}, \beta$, for any $D \sim \mathcal{N}(\mu, \sigma^2)^n$, if $n = \Omega(\frac{1}{\varepsilon_1}\log(\frac{1}{\beta}\log(\log(\frac{\sum_{i=1}^k\varepsilon_i + (n-k)\tau)}{\beta \tau})))$, Algorithm \ref{alg:EstimateRange} preserves $\boldsymbol{\varepsilon}$-PDP, and with probability at least $1 - \beta$, $\vert \widehat{\emph{\texttt{Range}}}(D)\vert = O(\sigma\sqrt{\log(\frac{\sum_{i=1}^k \varepsilon_i + (n-k)\tau}{\beta\tau})})$, and $\widehat{\emph{\texttt{Range}}}(D)$ covers $D$ except for at most $O(\frac{1}{\epsilon_1}\log(\frac{1}{\beta}\log(\log(\frac{\sum_{i=1}^k \varepsilon_i + (n-k)\tau}{\beta}))))$ elements.
\end{lemma}

\begin{proof}
    The privacy of Algorithm \ref{alg:EstimateRange} follows from Theorem \ref{theorem:diffusing_amplification} and basic composition. First, we explain the requirement on $n$: To satisfy the statement of Lemma \ref{lemma:range_estimation}, by \citep{Universal}, we need $\vert D_{\mathcal{F}}^{+} \vert = \Omega(\frac{1}{\tau}\log(\frac{1}{\beta}\frac{\texttt{Radius}(D_{\mathcal{F}})}{b}))$, where $\mu(\vert D_{\mathcal{F}}^{+} \vert) = \frac{\sum_{i=1}^k\varepsilon_i}{2\tau} + \frac{n - k}{2}$. To see how this is met, we use a case analysis on $k$:
    \begin{itemize}
        \item $k \leq \frac{n}{2}$: In this case, $n - k \geq \frac{n}{2}$, therefore $n = \Omega(\frac{1}{\varepsilon_1}\log(\frac{1}{\beta}\frac{\texttt{Radius}(D_{\mathcal{F}})}{b}))$ suffices to satisfy $\vert D_{\mathcal{F}}^{+} \vert = \Omega(\frac{1}{\tau}\log(\frac{1}{\beta}\frac{\texttt{Radius}(D_{\mathcal{F}})}{b}))$;

        \item $k > \frac{n}{2}$: In this case, $\frac{\sum_{i=1}^k\varepsilon_i}{\tau} \geq \frac{n}{2} \frac{\varepsilon_1}{\tau}$, therefore $n = \Omega(\frac{1}{\varepsilon_1}\log(\frac{1}{\beta}\frac{\texttt{Radius}(D_{\mathcal{F}})}{b}))$ shows $\vert D_{\mathcal{F}}^{+} \vert = \Omega(\frac{1}{\tau}\log(\frac{1}{\beta}\frac{\texttt{Radius}(D_{\mathcal{F}})}{b}))$.
    \end{itemize}
    Therefore, by \citep{Universal}, with probability at least $1 - \beta/2$, $\vert \widehat{\texttt{Range}}(D)\vert \leq 4 \vert \texttt{Range}(D_{\mathcal{F}})\vert + 6b$, and $\widehat{\texttt{Range}}(D)$ covers $D_{\mathcal{F}}^{+}$ except for at most $O(\frac{1}{\tau}\log(\frac{1}{\beta}\log(\frac{\texttt{Radius}(D_{\mathcal{F}}^{c})}{b})))$ elements. By standard Gaussian tail bound and the union bound, since $\vert\texttt{Range}(D_{\mathcal{F}})\vert\leq \vert\texttt{Range}(D))\vert$, for any $t > 0$, we have $\Pr(\vert\texttt{Range}(D_{\mathcal{F}})\vert \geq t) \leq (\sum_{i=1}^k \frac{\varepsilon_i}{\tau} + n-k) \exp(-\frac{t^2}{8\sigma^2})$, therefore, $\vert\texttt{Range}(D_{\mathcal{F}})\vert = O(\sigma\sqrt{\log(\frac{\sum_{i=1}^k \frac{\varepsilon_i}{\tau} + n-k}{\beta})})$, which further yields $\vert \widehat{\texttt{Range}}(D)\vert = O(\sigma\sqrt{\log(\frac{\sum_{i=1}^k \frac{\varepsilon_i}{\tau} + n-k}{\beta})})$, with probability at least $1 - \beta / 2$.

    For the number of outliers, We can partition the outliers into two parts: ones that are in $D_{\mathcal{F}}$ and ones that are not.

    For the first part, by Lemma \ref{lemma:bucketsize}, $\vert \widehat{\texttt{Range}}(D_{\mathcal{F}}) \vert = O(\sigma\sqrt{\log(\frac{\sum_{i=1}^k\varepsilon_i + (n-k)\tau)}{\beta \tau}})$, and $b = \Theta(\sigma)$ hold with probability at least $1 - \beta / 20$ and $1 - \beta / 20$, respectively. Therefore, by Chernoff's inequality, with probability at least $1 - 3\beta / 20$, there are no more than
    \begin{align}
    O\left( \frac{1}{\tau}\log(\frac{1}{\beta}\log(\sqrt{\log(\frac{\sum_{i=1}^k\varepsilon_i + (n-k)\tau)}{\beta \tau}}))\right) \nonumber
    \end{align}
    outliers on $D_{\mathcal{F}}$.

    For the second part, using Chernoff's inequality, with probability at least $1 - \beta / 20$, the number of the samples is $O(\sum_{i=1}^k(1 - \frac{\varepsilon_i}{2\tau}) + \frac{n-k}{2})$. Denote the number of outliers on $D_{\mathcal{F}} = \frac{c_1}{\tau}\log(\frac{1}{\beta}\log(\sqrt{\log(\frac{\sum_{i=1}^k\varepsilon_i + (n-k)\tau)}{\beta \tau}}))$, and define $\zeta := \frac{4c_1}{\sum_{i=1}^k\varepsilon_i + (n-k)\tau}\log(\frac{1}{\beta}\log(\log(\frac{\sum_{i=1}^k\varepsilon_i + (n-k)\tau)}{\beta\tau}))$. By setting the constant term of $n$ large enough, we have $\zeta < 0.5$. By Chernoff's inequality, with probability at least $1 - \beta / 20$, $\vert D_{\mathcal{F}}^{+} \cap (-\infty, F^{-1}(\zeta)] \vert \geq \zeta \vert D_{\mathcal{F}}^{+} \vert/2$. Also, with probability at least $1 - \beta / 20$, $\vert D_{\mathcal{F}}^{+} \cap [F^{-1}(1 - \zeta), \infty] \vert \geq \zeta \vert D_{\mathcal{F}}^{+} \vert/2$. Therefore, plug $\zeta$ in these two inequalities and by the number of outliers on $D_{\mathcal{F}}$, we have $[F^{-1}(\zeta), F^{-1}(1 - \zeta)] \subseteq \widehat{\texttt{Range}}(D_{\mathcal{F}})$ holds with probability at least $1 - 5 \beta / 20$. By Chernoff's inequality, $\vert (D - D_{\mathcal{F}}^{+}) \cap (-\infty, F^{-1}(\zeta)] \vert \leq 3 \zeta \vert D - D_{\mathcal{F}}^{+} \vert/2$ and $\vert (D - D_{\mathcal{F}}^{+}) \cap [F^{-1}(1 - \zeta), \infty) \vert \leq 3 \zeta \vert D - D_{\mathcal{F}}^{+} \vert/2$ hold with probability at least $1 - \beta/16$, respectively. Therefore, $\vert (D - D_{\mathcal{F}}^{+}) \cap \widehat{\texttt{Range}}(D_{\mathcal{F}})^c \vert \leq \vert (D - D_{\mathcal{F}}^{+}) \cap (-\infty, F^{-1}(\zeta)] \vert + \vert [F^{-1}(1 - \zeta), \infty) \cap (D - D_{\mathcal{F}}^{+}) \vert = 3 \zeta \vert D - D_{\mathcal{F}}^{+} \vert$. Thus, plug $\zeta$ in, the number of outliers in the diffused samples is no more than
    \begin{align*}
        O\left( \frac{\sum_{i=1}^k(1 - \frac{\varepsilon_i}{2\tau}) + \frac{n-k}{2}}{\sum_{i=1}^k\varepsilon_i + (n-k)\tau}\log(\frac{1}{\beta}\log(\log(\frac{\sum_{i=1}^k\varepsilon_i + (n-k)\tau)}{\beta}))\right),
    \end{align*}
    with probability at least $1 - 7 \beta / 20$.

    Combining two parts, with probability at least $1 - \beta$, $\vert \widehat{\texttt{Range}}(D)\vert = O(\sigma\sqrt{\log(\frac{\sum_{i=1}^k \varepsilon_i + (n-k)\tau}{\beta\tau})})$, and $\widehat{\texttt{Range}}(D)$ covers $D$ except for at most $O(\frac{1}{\epsilon_1}\log(\frac{1}{\beta}\log(\log(\frac{\sum_{i=1}^k \varepsilon_i + (n-k)\tau}{\beta}))))$ elements.
\end{proof}

\subsection{PDP Mean Estimation}
\label{section:mean_estimation_bounded}
With the estimated range $\widehat{\texttt{Range}}(D)$, we are able to perform a clipping upon the elements in $D$. Now, the clipped data have a range with a length well bounded and compared with the original distribution, now the data over this range becomes relatively uniformly distributed, which enables us to further utilize ADPM \citep{PDP_2023} to perform PDP mean estimation over the clipped dataset. Above all, we first invoke the PDP range estimator and then call ADPM over the clipped dataset. The overall error bound is in Theorem \ref{theorem:upper_bound_bounded}. Later we will show this error is optimal up to logarithmic factors.

\begin{minipage}[t]{0.93\linewidth}
    \centering
    \begin{algorithm}[H]
        \caption{PDPMeanEstimation}
        \label{alg:PDPMeanEstimation}
        \KwIn{$D$, $\boldsymbol{{\varepsilon}}$, $\beta$}
        \KwOut{$\widehat{\texttt{Mean}}(D)$}
        $\widehat{\texttt{Range}}(D) \gets \texttt{PDPEstimateRange}(D, \frac{\boldsymbol{\varepsilon}}{2}, \frac{\beta}{6})$\;
        $\widehat{\texttt{Mean}}(D) \gets \texttt{ADPM}(\texttt{Clipped}(D, \widehat{\texttt{Range}}(D)), \frac{\boldsymbol{\varepsilon}}{2}, \vert\widehat{\texttt{Range}}(D)\vert)$\;
        \Return{}~$\widehat{\texttt{Mean}}(D)$
        \end{algorithm}
    \end{minipage}

Now we present the main theorem of this section, which later shows that Algorithm \ref{alg:PDPMeanEstimation} achieves near-optimality:
\begin{theorem}
    Given $\boldsymbol{\varepsilon} \leq 1, \beta$, for any $D \sim \mathcal{N}(\mu, \sigma^2)^n$, if $n = \Omega(\frac{1}{\varepsilon_1}\log(\frac{1}{\beta}\log(\log(\frac{\sum_{i=1}^k\varepsilon_i + (n-k)\tau)}{\beta}))\log(\frac{1}{\beta}))$, Algorithm \ref{alg:PDPMeanEstimation} preserves $\boldsymbol{\varepsilon}$-PDP, and with probability at least $1-\beta$, it returns $\widehat{\emph{\texttt{Mean}}}(D)$ such that
    \begin{align*}
        \vert \widehat{\emph{\texttt{Mean}}}(D) - \mu \vert \leq & O\left( \frac{\sigma \log(\sum_{i=1}^k \varepsilon_i + (n-k)\tau)}{\sum_{i=1}^k\varepsilon_i + (n-k)\tau}\log(\frac{1}{\beta}\log(\log(\frac{\sum_{i=1}^k\varepsilon_i + (n-k)\tau)}{\beta\tau}))) \right. \\
        & \left. + \frac{\sigma\sqrt{\sum_{i=1}^k\varepsilon_i^2 + (n-k)\tau^2 + 8}}{\sum_{i=1}^k\varepsilon_i + (n-k)\tau}\log(\sum_{i=1}^k \varepsilon_i + (n-k)\tau)\log(\frac{1}{\beta})\right),
    \end{align*}
    where $k$ is the first index satisfying $\varepsilon_{k+1} \geq \frac{\sum_{i=1}^k \varepsilon_i^2 + 8}{\sum_{i=1}^k\varepsilon_i}$ and $\tau = \frac{\sum_{i=1}^k \varepsilon_i^2 + 8}{\sum_{i=1}^k\varepsilon_i}$. 
    \label{theorem:upper_bound_bounded}
\end{theorem}

\begin{proof}

    The privacy of Algorithm \ref{alg:PDPMeanEstimation} follows from Theorem \ref{theorem:diffusing_amplification}, Lemma \ref{lemma:privacy_weighted_mean}, and basic composition. We now analyze its utility.
    First, we introduce an auxiliary distribution $\mathcal{N}(\mu, \sigma^2)_{\xi}$, which is defined as $\mathcal{N}(\mu, \sigma^2)$ being truncated into $[\mu - \xi, \mu + \xi]$. Let $\xi = c\sigma\log(\sum_{i=1}^k \varepsilon_i + (n-k)\tau)$, where $c$ is a universal constant. By Lemma \ref{lemma:range_estimation}, with probability at least $1 - \beta/\text{12}$, by setting $c$ large enough, we have $\widehat{\texttt{Range}}(D_{\mathcal{F}}) \subseteq [\mu - \xi, \mu + \xi]$. Therefore, the error can be decomposed as
    \begin{equation}
        \begin{aligned}
            & \vert \hat{\mu}_{\texttt{ADPM}}(\texttt{Clipped}(D, \widehat{\texttt{Range}}(D_{\mathcal{F}}))) - \mu \vert ~(\text{by the triangle inequality})\\
            \leq & \vert \hat{\mu}_{\texttt{ADPM}}(\texttt{Clipped}(D, \widehat{\texttt{Range}}(D_{\mathcal{F}}))) - \mu_{\mathcal{N}(\mu, \sigma^2)_{\xi}}\vert +  \vert \mu_{\mathcal{N}(\mu, \sigma^2)_{\xi}} - \mu \vert.
        \end{aligned}
        \label{eq:error_dec_1}
    \end{equation}
    The second term in \eqref{eq:error_dec_1} is $0$, since $\mathcal{N}(\mu, \sigma^2)$ of our interest is symmetric. For the first term, we further define $D_{\xi}$ as $\mathcal{D}$ being truncated into $[\mu - \xi, \mu + \xi]$, which yields
    \begin{equation}
        \begin{aligned}
            & \vert \hat{\mu}_{\texttt{ADPM}}(\texttt{Clipped}(D, \widehat{\texttt{Range}}(D_{\mathcal{F}}))) - \mu_{\mathcal{N}(\mu, \sigma^2)_{\xi}}\vert \\
            \leq & \vert \hat{\mu}_{\texttt{ADPM}}(\texttt{Clipped}(D, \widehat{\texttt{Range}}(D_{\mathcal{F}}))) - \mu_{\texttt{ADPM}}(D_{\xi})\vert + \vert \mu_{\texttt{ADPM}}(D_{\xi}) - \mu_{\mathcal{N}(\mu, \sigma^2)_{\xi}} \vert.
        \end{aligned}
        \label{eq:error_dec_2}
    \end{equation}

    We first analyze the first term in \eqref{eq:error_dec_2}:
    \begin{equation}
        \begin{aligned}
            & \vert \hat{\mu}_{\texttt{ADPM}}(\texttt{Clipped}(D, \widehat{\texttt{Range}}(D_{\mathcal{F}}))) - \mu_{\texttt{ADPM}}(D_{\xi})\vert \\
            \leq & \vert {\mu}_{\texttt{ADPM}}(\texttt{Clipped}(D_{\xi}, \widehat{\texttt{Range}}(D_{\mathcal{F}}))) - \mu_{\texttt{ADPM}}(D_{\xi}) \vert + \vert \text{Lap}(\frac{8\vert\widehat{\texttt{Range}}(D_{\mathcal{F}})\vert}{\sum_{i=1}^k\varepsilon_i + (n-k)\tau})\vert.
        \end{aligned}
        \label{eq:error_dec_3}
    \end{equation}

    First, we bound the first term in \eqref{eq:error_dec_3}. Since $\varepsilon_i$ and $x_i$ are independent, the contribution of each outlier in mean estimation is bounded by $\frac{2\xi E_j}{\sum_{i=1}^k\varepsilon_i + (n-k)\tau}$, where $E_j$ is randomly drawn from $\boldsymbol{\varepsilon}$ (note that $\frac{E_j}{\sum_{i=1}^k\varepsilon_i + (n-k)\tau}$ is the weight of ADPM).

    Here, one needs to be careful about $E_j$, since it can be subtle to determine its distribution. As Lemma \ref{lemma:range_estimation}, we can partition the outliers into two parts: ones that are not diffused and ones that are diffused, to decouple the correlation.

    Here, a finer analysis is needed, since Bernstein type bound cannot be used directly. More specifically, we cannot directly claim that the $\varepsilon$ of an outlier in or not in the diffused dataset follows a discrete distribution: Since it is sampled from a Poisson random set. Also, the elements that sampled from a set without replacement are not independent. To address these two points, we further decompose it into two parts.

    For the first part, by Lemma \ref{lemma:range_estimation}, with probability at least $1 - 3\beta / \text{12}$, there are no more than
    \begin{align}
        n_1 & := O\left( \frac{1}{\tau}\log(\frac{1}{\beta}\log(\frac{\sigma}{b}\sqrt{\log(\frac{\sum_{i=1}^k\varepsilon_i + (n-k)\tau)}{\beta}}))\right) \nonumber \\ 
        & := \frac{c_1}{\tau}\log(\frac{1}{\beta}\log(\log(\frac{\sum_{i=1}^k\varepsilon_i + (n-k)\tau)}{\beta}))
        \label{eq:number_outlier_part_1}
    \end{align}
    outliers on $D_{\mathcal{F}}$, where $c_1$ is a universal constant. Since for any element, its data and privacy budget are independent, therefore, when we aim to control the sum of these $\varepsilon_i$s on the diffused dataset, it is equivalent to sample $n_1$ elements on the diffused dataset without replacement. Hence, this is a two-stage sampling scheme, i.e., Poisson sampling plus subsampling without replacement, and we aim to bound the sum of the subsampled set. This can be achieved by applying Lemma \ref{lemma:two-stage_sampling_bound}. Here we do some preparations: $t=\sqrt{\sum_{i=1}^k\varepsilon_i^2 + (n-k)\tau^2 + 8}$, $\boldsymbol{\tilde{\varepsilon}}$ is equal to $\boldsymbol{\varepsilon}$ except for the last $n - k$ entries being $\tau$, $\mu_N^{(1)} = \frac{\sum_{i=1}^{k}\varepsilon_i + (n-k)\tau}{2\tau}, \mu_T^{(1)} = \frac{\sum_{i=1}^k \frac{\varepsilon_i^2}{\tau} + (n-k)\tau}{\sum_{i=1}^k \frac{\varepsilon_i}{\tau} + (n-k)}, \eta^{(1)} = \frac{n}{2\mu_N^{(1)}}, \gamma(A) = \tau - \varepsilon_1, \sigma^2(A) = \sigma^2(\tilde{\boldsymbol{\varepsilon}})$. Denote the sum of the outliers' privacy budgets as $T_1$, applying Lemma \ref{lemma:two-stage_sampling_bound}:
    \begin{align*}
        \Pr(T_1 - n_1 \mu_T^{(1)} \geq t) & \leq 
        \exp(-\frac{t^2}{\frac{32n_1 n \sigma^2(\tilde{\boldsymbol{\varepsilon}})\tau}{\sum_{i=1}^{k}\varepsilon_i + (n-k)\tau} + \frac{8}{3}\tau t})
        + 
        \exp(-\frac{t^2\frac{\sum_{i=1}^{k}\varepsilon_i + (n-k)\tau}{\tau}}{16n_1^2\frac{n\tau}{\sum_{i=1}^{k}\varepsilon_i + (n-k)\tau}\sigma^2(\tilde{\boldsymbol{\varepsilon}}) + \frac{16}{3}n_1\tau t}) \\
        & + 
        \exp(-\frac{\sum_{i=1}^{k}\varepsilon_i + (n-k)\tau}{4\tau}) \\
        & := \exp(-\frac{A_1}{A_2 + A_3}) + \exp(-\frac{B_1}{B_2 + B_3}) + \exp(-C_1)
    \end{align*}

    For $A_1 /A_2$:
    \begin{align*}
        \frac{A_1}{A_2} = \frac{\sum_{i=1}^{k}\varepsilon_i^2 + (n-k)\tau^2 + 8}{\frac{32n_1 n \sigma^2(\tilde{\boldsymbol{\varepsilon}})\tau}{\sum_{i=1}^{k}\varepsilon_i + (n-k)\tau}} \geq \frac{\sigma^2(\tilde{\boldsymbol{\varepsilon}})n}{\frac{32n_1 n \sigma^2(\tilde{\boldsymbol{\varepsilon}})\tau}{\sum_{i=1}^{k}\varepsilon_i + (n-k)\tau}}
        =
        \frac{\sum_{i=1}^{k}\varepsilon_i + (n-k)\tau}{32n_1\tau} := \Omega(\log(\frac{1}{\beta})).
    \end{align*}

    For $A_1 / A_3$ we use a case analysis:
    
    1) $\tilde{\boldsymbol{\varepsilon}} = \boldsymbol{\varepsilon}$: This case indicates that there is no saturation. Therefore, we can use the relation $\varepsilon_{n} < \frac{\sum_{i=1}^{n-1} \varepsilon_i^2 + 8}{\sum_{i=1}^{n-1}\varepsilon_i}$ and plug it in yields
    \begin{align*}
        \frac{A_1}{A_3} = \frac{3}{8\varepsilon_n}\sqrt{\sum_{i=1}^n\varepsilon_i^2 + 8} \geq \frac{1}{8\varepsilon_n} \sqrt{\varepsilon_n^2 + \varepsilon_n\sum_{i=1}^{n-1}\varepsilon_i} = \frac{1}{8} \sqrt{\frac{\sum_{i=1}^n \varepsilon_i}{\varepsilon_n}} \geq \frac{1}{8} \sqrt{\frac{n\varepsilon_1}{\varepsilon_n}} = \Omega(\log(\frac{1}{\beta})).
    \end{align*}

    2) $\tilde{\boldsymbol{\varepsilon}} \neq \boldsymbol{\varepsilon}$: In this case
    \begin{align*}
        \frac{A_1}{A_3} = \frac{3}{8\tau}\sqrt{\sum_{i=1}^k\varepsilon_i^2 + (n-k)\tau^2 + 8}.
    \end{align*}
    We further consider a) $k \leq \frac{n}{2}$: then $\frac{A_1}{A_3} = \Omega(\log(\frac{1}{\beta})$;
    b) $k > \frac{n}{2}$: We use the relation $\tau = \frac{\sum_{i=1}^k\varepsilon_i^2 + 8}{\sum_{i=1}^k \varepsilon_i}$ and plug it in:
    \begin{align*}
        \frac{A_1}{A_3} \geq \frac{1}{8} \frac{\sum_{i=1}^k \varepsilon_i}{\sum_{i=1}^k\varepsilon_i^2 + 8}\sqrt{\sum_{i=1}^k\varepsilon_i^2 + 8} \geq \frac{1}{9} \frac{\sum_{i=1}^k \varepsilon_i}{\sqrt{\sum_{i=1}^n \varepsilon_i^2 + 8}} = \frac{1}{9} \sqrt{\frac{\sum_{i=1}^k \varepsilon_i}{\tau}} \geq \frac{1}{9}\sqrt{\sum_{i=1}^{\frac{n}{2}}\varepsilon_i} = \Omega(\log(\frac{1}{\beta}).
    \end{align*}

    For $B_1 / B_2$:

    \begin{align*}
        \frac{B_1}{B_2} = \frac{\sum_{i=1}^k\varepsilon_i^2 + (n-k)\tau^2 + 8}{16n_1^2 n\sigma^2(\tilde{\boldsymbol{\varepsilon}})}\frac{(\sum_{i=1}^k\varepsilon_i + (n-k)\tau)^2}{\tau^2} \geq \frac{(\sum_{i=1}^k\varepsilon_i + (n-k)\tau)^2}{32n_1^2\tau^2} \overset{\text{Plug}~n_1~\text{in}}{=} \Omega(\log^2(\frac{1}{\beta})).
    \end{align*}

    For $B_1 / B_3$:
    \begin{align*}
        \frac{B_1}{B_3} \geq \frac{t}{8\tau}\frac{\sum_{i=1}^k\varepsilon_i + (n-k)\tau}{n_1\tau} = \Omega(\log(\frac{1}{\beta})).
    \end{align*}

    For $C_1$:
    \begin{align*}
        C_1 = \frac{\sum_{i=1}^k\varepsilon_i + (n-k)\tau}{4\tau} \geq \frac{n\varepsilon_1}{\tau} = \Omega(\log(\frac{1}{\beta})).
    \end{align*}

    Therefore, $T_1$ is concentrated around $n_1\mu_{T}^{(1)}$.
        
    For the second part, the number of outliers in the diffused samples is no more than
    \begin{align*}
        n_2 := O\left( \frac{\sum_{i=1}^k(1 - \frac{\varepsilon_i}{2\tau}) + \frac{n-k}{2}}{\sum_{i=1}^k\varepsilon_i + (n-k)\tau}\log(\frac{1}{\beta}\log(\log(\frac{\sum_{i=1}^k\varepsilon_i + (n-k)\tau)}{\beta}))\right),
    \end{align*}
    with probability at least $1 - 7 \beta / \text{12}$. Denote $\mu_N^{(2)} = \sum_{i=1}^k(1 - \frac{\varepsilon_i}{2\tau}) + \frac{n-k}{2}, \mu_T^{(2)} = \frac{\sum_{i=1}^k(1 - \frac{\varepsilon_i}{2\tau})\varepsilon_i + \frac{(n-k)\tau}{2}}{\sum_{i=1}^k(1 - \frac{\varepsilon_i}{2\tau}) + \frac{n-k}{2}}, \eta^{(2)} = \frac{n(1 - \frac{\varepsilon_1}{\tau})}{\mu_N^{(2)}}, \gamma(A) = \tau - \varepsilon_1, \sigma^2(A) = \sigma^2(\tilde{\boldsymbol{\varepsilon}})$. Denote the sum of the outliers' privacy budgets as $T_2$, applying Lemma \ref{lemma:two-stage_sampling_bound}:
    \begin{align*}
        \Pr(T_2 - n_2 \mu_T^{(2)} \geq t) & \leq 
        \exp(-\frac{t^2}{\frac{16n_2 n \sigma^2(\tilde{\boldsymbol{\varepsilon}})}{\mu_N^{(2)}} + \frac{8}{3}\tau t})
        + 
        \exp(-\frac{t^2\mu_N^{(2)}}{32n_2^2\frac{n(1 - \frac{\varepsilon_1}{\tau})}{\mu_N^{(2)}}\sigma^2(\tilde{\boldsymbol{\varepsilon}}) + \frac{16}{3}n_2\tau t}) \\
        & + 
        \exp(-(\sum_{i=1}^k(1 - \frac{\varepsilon_i}{2\tau}) + \frac{n-k}{2})) \\
        & := \exp(-\frac{A_4}{A_5 + A_6}) + \exp(-\frac{B_4}{B_5 + B_6}) + \exp(-C_2).
    \end{align*}

    For $A_4 / A_5$:
    \begin{align*}
        \frac{A_4}{A_5} \geq \frac{\mu_N^{(2)}}{16 n_2} = O(\frac{A_1}{A_3}) = \Omega(\log(\frac{1}\beta)).
    \end{align*}

    For $A_4 / A_6$:
    \begin{align*}
        \frac{A_4}{A_6} = \frac{A_1}{A_3} = \Omega(\log(\frac{1}{\beta})).
    \end{align*}

    For $B_4/B_5$:
    \begin{align*}
        \frac{B_4}{B_6} \geq \frac{(\mu_N^{(2)})^2}{32n_2^2(1 - \frac{\varepsilon_1}{\tau})} = \Omega(\log(\frac{1}{\beta})).
    \end{align*}

    For $B_4/B_6$:
    \begin{align*}
        \frac{B_4}{B_6} \geq O(\frac{B_1}{B_3}) = \Omega(\log(\frac{1}{\beta})).
    \end{align*}
    Therefore, $T_2$ is concentrated around $n_2\mu_{T}^{(2)}$.

     With probability at least $1 - \beta/\text{12}$, respectively, the clipping errors of the two parts are bounded as
    \begin{equation}
        \begin{aligned}
            & \sum_{j=1}^{n_1} \frac{2\xi T_1}{\sum_{i=1}^k\varepsilon_i + (n-k)\tau} \\ 
            \leq & O\left(\frac{2\xi}{\sum_{i=1}^k\varepsilon_i + (n-k)\tau} \left(\frac{\sum_{i=1}^k\varepsilon_i^2 + (n-k)\tau^2}{\sum_{i=1}^k\varepsilon_i + (n-k)\tau} \frac{1}{\tau}\right) \log(\frac{1}{\beta}\log(\log(\frac{\sum_{i=1}^k\varepsilon_i + (n-k)\tau)}{\beta \tau}))) \right)\\
            + & O\left(\frac{2\xi}{\sum_{i=1}^k\varepsilon_i + (n-k)\tau}\sqrt{\sum_{i=1}^k\varepsilon_i^2 + (n-k)\tau^2 + 8} \right) \\
            \leq & O\left(\frac{2\xi}{\sum_{i=1}^k\varepsilon_i + (n-k)\tau} \log(\frac{1}{\beta}\log(\log(\frac{\sum_{i=1}^k\varepsilon_i + (n-k)\tau)}{\beta \tau}))) \right)\\
            + & O\left(\frac{2\xi}{\sum_{i=1}^k\varepsilon_i + (n-k)\tau}\sqrt{\sum_{i=1}^k\varepsilon_i^2 + (n-k)\tau^2 + 8}\right),
        \end{aligned}
        \label{eq:clipping_error_1}
    \end{equation}

        \begin{equation}
            \begin{aligned}
                & \sum_{j=1}^{n_2} \frac{2\xi T_2}{\sum_{i=1}^k\varepsilon_i + (n-k)\tau} \\ 
                \leq & O\left(\frac{2\xi}{\sum_{i=1}^k\varepsilon_i + (n-k)\tau} \left(\frac{\sum_{i=1}^k(1 - \frac{\varepsilon_i}{2\tau})\varepsilon_i + \frac{(n-k)\tau}{2}}{\sum_{i=1}^k\varepsilon_i + (n-k)\tau}\right) \log(\frac{1}{\beta}\log(\log(\frac{\sum_{i=1}^k\varepsilon_i + (n-k)\tau)}{\beta \tau}))) \right)\\
                + & O\left(\frac{2\xi}{\sum_{i=1}^k\varepsilon_i + (n-k)\tau}\sqrt{\sum_{i=1}^k\varepsilon_i^2 + (n-k)\tau^2 + 8} \right)\\
                \leq & O\left(\frac{2\xi}{\sum_{i=1}^k\varepsilon_i + (n-k)\tau} \log(\frac{1}{\beta}\log(\log(\frac{\sum_{i=1}^k\varepsilon_i + (n-k)\tau)}{\beta \tau}))) \right)\\
                + & O\left(\frac{2\xi}{\sum_{i=1}^k\varepsilon_i + (n-k)\tau}\sqrt{\sum_{i=1}^k\varepsilon_i^2 + (n-k)\tau^2 + 8}\right).
            \end{aligned}
            \label{eq:clipping_error_2}
        \end{equation}

        By \eqref{eq:clipping_error_1} and \eqref{eq:clipping_error_2}, with probability at least $1 - 10 \beta/\text{12}$, the total clipping error is bounded as 
        \begin{equation}
            \begin{aligned}
                & O\left(\frac{2\xi}{\sum_{i=1}^k\varepsilon_i + (n-k)\tau} \log(\frac{1}{\beta}\log(\log(\frac{\sum_{i=1}^k\varepsilon_i + (n-k)\tau)}{\beta \tau}))) \right) + \\
                & O\left(\frac{2\xi}{\sum_{i=1}^k\varepsilon_i + (n-k)\tau}\sqrt{\sum_{i=1}^k\varepsilon_i^2 + (n-k)\tau^2 + 8}\right).
            \end{aligned}
            \label{eq:bound_clipping_error}
        \end{equation}
        Please note that $\xi = O(\sigma\log(\sum_{i=1}^k \varepsilon_i + (n-k)\tau))$ for Gaussian distribution.

        The added noise, i.e., the second term in \eqref{eq:error_dec_3}, can be bounded using the Laplace tail bound: With probability at least $1 - \beta/\text{12}$, 
        \begin{align}
            \vert \text{Lap}(\frac{8\vert\widehat{\texttt{Range}}(D_{\mathcal{F}})\vert}{\sum_{i=1}^k\varepsilon_i + (n-k)\tau})\vert \leq \frac{8\log(\frac{\text{12}}{\beta})\vert\widehat{\texttt{Range}}(D_{\mathcal{F}})\vert}{\sum_{i=1}^k\varepsilon_i + (n-k)\tau} = O(\log(\frac{1}{\beta})\frac{\sigma\sqrt{\log(\sum_{i=1}^k\varepsilon_i + (n-k)\tau})}{\sum_{i=1}^k\varepsilon_i + (n-k)\tau}).
            \label{eq:bound_noise}
        \end{align}
        
        The second term in $\eqref{eq:error_dec_2}$ is the sampling error of noiseless ADPM, which can be bounded using the weighted Bernstein's inequality (Theorem \ref{theorem:weighted_Bernstein}):

        With probability at least $1 - \beta/\text{12}$, the sampling error is bounded as 
        \begin{align}
            O(\frac{\sigma\sqrt{\sum_{i=1}^{k}\varepsilon_i^2 + (n-k)\tau^2}}{(\sum_{i=1}^{k}\varepsilon_i + (n-k)\tau)}\sqrt{\log(\frac{1}{\beta})}) + O(\frac{\tau\xi}{\sum_{i=1}^k \varepsilon_i + (n-k)\tau}\log(\frac{1}{\beta})).
            \label{eq:bound_sampling_error}
        \end{align}

        Combining \eqref{eq:bound_clipping_error}, \eqref{eq:bound_noise}, and \eqref{eq:bound_sampling_error} bounds the estimation error in \eqref{eq:error_dec_1}.

\end{proof}

\begin{lemma}[Concentration for Two Stage Sampling]
\label{lemma:two-stage_sampling_bound}
Given sampling probabilities $p_1, p_2,\hdots,p_n$, for any set $A = \{a_1,\dots,a_n\}, a_i\in [0,1]$ with $\gamma(A)$ and $\sigma^2(A)$ as its width and population variance, respectively. Define
\begin{align*}
\mu_N := \sum_{i=1}^n p_i, \mu_T := \frac{\sum_{i=1}^n p_i a_i}{\mu_N}, \text{and}~\eta = \max_i \frac{n p_i}{\mu_N}.
\end{align*}

Consider the following two--stage sampling procedure:
\begin{enumerate}
    \item \emph{\textbf{Poisson sampling}} Independently draw
    \begin{align*}
    I_i \sim \ \mathrm{Bernoulli}(p_i),\quad i=1,\dots,n,
    \end{align*}
    and let $S := \{i : I_i = 1\}$ and $N := |S|$.
    \item \emph{\textbf{Without-replacement subsampling}} On the event $\{N \ge m\}$, sample $m$ indices
    without replacement uniformly from $S$. Let the corresponding values be
    $Y_1,\dots,Y_m\in[0,1]$, and define
    \begin{align*}
    T := \sum_{k=1}^m Y_k.
    \end{align*}
\end{enumerate}

Then the following holds:
for any $t > 0$,
\begin{align*}
\Pr(
T - m\mu_T
\geq
t) 
\leq
\exp\big( - \frac{t^2}{\frac{16 m n \sigma^2(A)}{\mu_{N}} + \frac{8}{3}\gamma(A)t}\big)
+
\exp\big(
-\frac{t^2\mu_N}{
32 m^2 \eta \sigma^2(A) + \frac{16}{3} \gamma(A) m t}
\big)
+
\exp\left(-\frac{\mu_N}{4}\right)
\end{align*}
\end{lemma}

\begin{proof}
First, define $\mu_S := \frac{\sum_{i \in S} a_i}{N}$, we decompose the shift as:
\begin{align*}
    \Pr(T - m\mu_T \geq t) \leq \Pr(T - m\mu_S \geq \frac{t}{2}) + \Pr(m(\mu_S - \mu_T) \geq \frac{t}{2}),
\end{align*}
where the first term on RHS stands for the subsampling shift given $S$, and the second term stands for the concentration behavior of $\mu_S$ around $\mu_T$. We bound them in turn.

For the subsampling error, by the finite-population Bernstein inequality for sampling without replacement, we have
\begin{align*}
    \Pr(T - m\mu_S \geq \frac{t}{2})
    \leq
    \exp\big( - \frac{(\frac{t}{2})^2}{2m\sigma^2(S)\cdot(1 - \frac{m-1}{N}) + \frac{2}{3}\gamma(S)t}\big)
    \leq \exp\big( - \frac{t^2}{8m\sigma^2(S) + \frac{8}{3}\gamma(S)t}\big).
\end{align*}
Since $\sigma^2(S)$ and $\gamma(S)$ are $S$-dependent, while the desired bound is not. We relate them to statistics on $A$. Note that $\sigma^2(S) \leq \frac{2n}{\mu_N}\sigma^2(A)$ holds for any subset $\vert S \vert \geq \frac{\mu_N}{2}$ on $A$ and $\gamma(S) \leq \gamma(A)$, we have
\begin{align*}
\Pr(
T - m\mu_T
\geq
t) 
\leq
\exp\big( - \frac{t^2}{\frac{16mn\sigma^2(A)}{\mu_{N}} + \frac{8}{3}\gamma(A)t}\big) + \exp\big(-\frac{\mu_N}{8}\big).
\end{align*}

Next, we control the second term, using the technique of ratio concentration, which yields:
\begin{align*}
\Pr(
m\mu_S - m\mu_T
\geq
t) 
\leq
\exp\big(
-\frac{t^2\mu_N}{
32 m^2 \eta \sigma^2(A) + \frac{16}{3} \gamma(A) m t}
\big)
+
\exp\left(-\frac{\mu_N}{8}\right).
\end{align*}

Combining these two terms yields the control of the shift.

\end{proof}

\begin{lemma}[Ratio Concentration]
\label{lemma:ratio_concentration}
Let $A = \{a_1,\dots,a_n\}\in[L,U]$ and define $M := U-L$. Let
\begin{align*}
\bar a := \frac1n\sum_{i=1}^n a_i,
\qquad
\sigma^2(A) := \frac1n\sum_{i=1}^n (a_i-\bar a)^2.
\end{align*}
Let $p_1,\dots,p_n\in(0,1]$ be given sampling probabilities and set
\begin{align*}
\mu_N := \sum_{i=1}^n p_i,
\qquad
w_i := \frac{p_i}{\mu_N},\quad \sum_{i=1}^n w_i = 1,
\qquad
w_{\max} := \max_{1\le i\le n} w_i.
\end{align*}
Define the weighted mean
\begin{align*}
\mu_T := \sum_{i=1}^n w_i a_i = \frac{\sum_{i=1}^n p_i a_i}{\mu_N}.
\end{align*}
Consider Poisson sampling
\begin{align*}
I_i \sim \mathrm{Bernoulli}(p_i)\ \text{independently},\quad
S := \{i : I_i = 1\},\quad N := |S|.
\end{align*}
On the event $\{N\ge 1\}$, define the Poisson-sample mean
\begin{align*}
\mu_S := \frac{1}{N}\sum_{i\in S} a_i.
\end{align*}
Then, for any $t>0$,
\begin{align*}
\Pr(m(\mu_S - \mu_T) \ge t)
\le
\exp\left(
-\frac{t^2\mu_N}{
8m^2nw_{\max}\sigma^2(A) + \frac{4}{3} M m t}
\right)
+
\exp\left(-\frac{\mu_N}{8}\right).
\end{align*}
\end{lemma}

\begin{proof}
Define the weighted variance
\begin{align*}
\sigma_T^2 := \sum_{i=1}^n w_i (a_i-\mu_T)^2.
\end{align*}
Writing $a_i-\mu_T = (a_i-\bar a) + (\bar a-\mu_T)$ and expanding yields
\begin{align}
\sigma_T^2
= \sum_{i=1}^n w_i(a_i-\bar a)^2 - (\bar a-\mu_T)^2
\le \sum_{i=1}^n w_i(a_i-\bar a)^2
\le w_{\max} \sum_{i=1}^n (a_i-\bar a)^2
= n w_{\max}\sigma^2(A).
\label{eq:sigmaT_vara_bound}
\end{align}

Next, note that on $\{N\ge 1\}$,
\begin{align*}
\mu_S = \frac{1}{N}\sum_{i=1}^n I_i a_i.
\end{align*}
Let $b_i := a_i - \mu_T$. Since $\sum_i p_i b_i = 0$, we have
\begin{align*}
\mu_S - \mu_T
= \frac{1}{N}\sum_{i=1}^n I_i b_i
= \frac{1}{N}\sum_{i=1}^n (I_i-p_i) b_i.
\end{align*}
Define
\begin{align*}
Z_i := (I_i-p_i)b_i = (I_i-p_i)(a_i-\mu_T),\qquad
S := \sum_{i=1}^n Z_i.
\end{align*}
Then on $\{N\ge 1\}$,
\begin{equation}
\mu_S - \mu_T = \frac{S}{N}.
\label{eq:muS_muT_ratio}
\end{equation}

Introduce the event
\begin{align*}
\mathcal{E}_N := \left\{ N \ge \frac{\mu_N}{2} \right\}.
\end{align*}
On $\mathcal{E}_N$, if $m(\mu_S-\mu_T)\ge t$, then by \eqref{eq:muS_muT_ratio}
\begin{align*}
S = N(\mu_S-\mu_T) \ge \frac{\mu_N}{2}\cdot \frac{t}{m}
=: u.
\end{align*}
Therefore
\begin{align*}
\{m(\mu_S-\mu_T)\ge t,\ \mathcal{E}_N\}
\subset \{S\ge u\},
\end{align*}
and hence
\begin{equation}
\Pr(m(\mu_S-\mu_T)\ge t)
\le
\Pr(S\ge u) + \Pr(\mathcal{E}_N^c).
\label{eq:main_split}
\end{equation}

We now bound the two terms separately. First, note that $Z_i$ are independent, $\mathbb{E}Z_i=0$, and since $a_i\in[L,U]$ and $\mu_T$ is a convex combination of the $a_i$, we have $|a_i-\mu_T|\le M:=U-L$, thus $|Z_i|\le M$. Moreover,
\begin{align*}
\sigma^2(Z_i)
= (a_i-\mu_T)^2 \sigma^2(I_i-p_i)
= (a_i-\mu_T)^2 p_i(1-p_i)
\le p_i(a_i-\mu_T)^2.
\end{align*}
Hence
\begin{align*}
\sum_{i=1}^n \sigma^2(Z_i)
\le \sum_{i=1}^n p_i(a_i-\mu_T)^2
= \mu_N \sum_{i=1}^n w_i(a_i-\mu_T)^2
= \mu_N \sigma_T^2.
\end{align*}
Bernstein's inequality for sums of independent, zero-mean, bounded random variables then yields, for all $u>0$,
\begin{align*}
\Pr(S\ge u)
\le
\exp\left(
-\frac{u^2}{2\mu_N\sigma_T^2 + \tfrac{2}{3}Mu}
\right).
\end{align*}
Taking $u = t\mu_N/(2m)$ gives
\begin{align*}
\Pr(S\ge u)
\le
\exp\left(
-\frac{(t^2\mu_N^2)/(4m^2)}{2\mu_N\sigma_T^2 + \tfrac{2}{3}M\cdot \frac{t\mu_N}{2m}}
\right)
=
\exp\left(
-\frac{t^2\mu_N/(4m^2)}{2\sigma_T^2 + \frac{M t}{3m}}
\right).
\end{align*}
Using \eqref{eq:sigmaT_vara_bound}, we have $\sigma_T^2 \le n w_{\max}\sigma^2(A)$, hence
\begin{align*}
2\sigma_T^2 + \frac{M t}{3m}
\le 2n w_{\max}\sigma^2(A) + \frac{M t}{3m}.
\end{align*}
Therefore
\begin{align*}
\Pr(S\ge u)
\le
\exp\left(
-\frac{t^2\mu_N/(4m^2)}{2n w_{\max}\sigma^2(A) + \frac{M t}{3m}}
\right).
\end{align*}
Noting that
\begin{align*}
4m^2(2n w_{\max}\sigma^2(A) + \tfrac{M t}{3m})
=
8m^2 n w_{\max}\sigma^2(A) + \tfrac{4}{3}M m t,
\end{align*}
we obtain
\begin{equation}
\Pr(S\ge u)
\le
\exp\left(
-\frac{t^2\mu_N}{
8m^2nw_{\max}\sigma^2(A) + \frac{4}{3} M m t}
\right).
\label{eq:bern_S_final}
\end{equation}

Next, we bound $\Pr(\mathcal{E}_N^c)$. Since $N=\sum_{i=1}^n I_i$ is Poisson--Binomial with mean $\mu_N$, a standard Chernoff bound yields, for $\alpha\in(0,1)$,
\begin{align*}
\Pr(N \le (1-\alpha)\mu_N)
\le
\exp\left(-\frac{\alpha^2}{2}\mu_N\right).
\end{align*}
Taking $\alpha=1/2$ gives
\begin{equation}
\Pr(\mathcal{E}_N^c)
=
\Pr\left(N < \frac{\mu_N}{2}\right)
\le
\exp\left(-\frac{\mu_N}{8}\right).
\label{eq:chernoff_N}
\end{equation}

Combining \eqref{eq:main_split}, \eqref{eq:bern_S_final} and \eqref{eq:chernoff_N} yields
\begin{align*}
\Pr(m(\mu_S-\mu_T)\ge t)
\le
\exp\left(
-\frac{t^2\mu_N}{
8m^2n w_{\max}\sigma^2(A) + \frac{4}{3} M m t}
\right)
+
\exp\left(-\frac{\mu_N}{8}\right),
\end{align*}
which is the desired inequality.
\end{proof}

By comparing the upper bound in Theorem \ref{theorem:lower_bound} and the above upper bound, we can see that our upper bound for mean estimation matches the lower bound up to some logarithmic factors.

\section{PDP Mean Estimation Under the Unbounded PDP setting}

After addressing mean estimation under the bounded PDP setting, we move on to the more complex setting: unbounded PDP setting. As shown in \citep{add_remove_one_mean_estimation}, the unbounded PDP setting provides stronger privacy protection than the bounded PDP setting, as not only the content of each element but also its presence is considered private.  In Section \ref{section:lower_bound_unbounded}, we first establish a lower bound for this problem, and then propose a mean estimation algorithm whose upper bound matches the lower bound up to logarithmic factors in Section \ref{section:upper_bound_unbounded}.

\subsection{Privacy-Specific Lower Bound}
    \label{section:lower_bound_unbounded}
     As mentioned before, one key difference between the bounded PDP setting and the unbounded PDP setting is that, the privacy budgets also become private information and can vary among database instances. To establish the minimax lower bound under this setting, one idea is to derive a lower bound over all possible privacy vectors. However, such a lower bound can be vacuous and provides limited insight. On the other hand, we would like to derive a lower bound specific to the privacy vector of a dataset, which can be more insightful. To obtain this bound, we first establish a lemma relating the distance between two output distributions to the input distributions in the unbounded PDP setting, which serves the same role as Lemma \ref{lemma:PDP_2023_TV} in the bounded PDP setting:
     \begin{lemma}
        \label{lemma:add/remove-one_TV}
         Given a user universe $\mathcal{U}$ and the privacy function $\mathcal{E}$, we use $\boldsymbol{\varepsilon}_D$ to denote the (sorted) privacy vector of $D$ that contains all $\mathcal{E}(u)$ of $D$. For any $\mathcal{M}$ that satisfies $\mathcal{E}$-PDP and any distribution $\mathcal{P}_1, \mathcal{P}_2$, the following inequality holds for any $k \leq n$,
        \begin{align*}
            \Vert \mathcal{M}(\mathcal{P}_1^n)) - \mathcal{M}(\mathcal{P}_2^n)) \Vert_{\mathrm{TV}} \leq 2 \Vert \mathcal{P}_1 - \mathcal{P}_2 \Vert_{\mathrm{TV}} \sum_{i=1}^k (1 - e^{-{\varepsilon_D}_i}) + \sqrt{\frac{n-k}{2} D_{\mathrm{KL}}(\mathcal{P}_1 \Vert \mathcal{P}_2)},
        \end{align*}
     \end{lemma}

     \begin{proof}
         The proof mainly follows the proof of Lemma \ref{lemma:PDP_2023_TV} in \citep{PDP_2023}, with a difference in introducing PDP to the derivation.

         Let $\mathcal{M}(\mathcal{P}_1^k \mathcal{P}_2^{n-k})$ be the distribution of the output of the $\boldsymbol{\varepsilon}$-PDP estimator $\mathcal{M}(\cdot)$ when the input is $D$ drawn from the product distribution $\mathcal{P}_1^kP_2^{n-k}$.

        By the triangle inequality,
        \begin{align*}
            \Vert\mathcal{M}(\mathcal{P}_1^n)-\mathcal{M}(\mathcal{P}_2^n)\Vert_{\mathrm{TV}}
            \le \Vert\mathcal{M}(\mathcal{P}_1^n)-\mathcal{M}(\mathcal{P}_1^k \mathcal{P}_2^{n-k})\Vert_{\mathrm{TV}}
            + \Vert\mathcal{M}(\mathcal{P}_1^k \mathcal{P}_2^{n-k})-\mathcal{M}(\mathcal{P}_2^n)\Vert_{\mathrm{TV}}.
        \end{align*}
        
        By the Data Processing Inequality and Pinsker’s inequality,
        \begin{align*}
            \Vert\mathcal{M}(\mathcal{P}_1^n)-\mathcal{M}(\mathcal{P}_1^k \mathcal{P}_2^{n-k})\Vert_{\mathrm{TV}}
            \le \Vert \mathcal{P}_1^n - \mathcal{P}_1^k \mathcal{P}_2^{n-k} \Vert_{\mathrm{TV}}
            \le \sqrt{\frac{n-k}{2}D_{\mathrm{KL}}(\mathcal{P}_1 \Vert \mathcal{P}_2)}.
        \end{align*}
        
        On event $A$, we have
        \begin{align*}
            \vert
            \mathcal{M}(\mathcal{P}_1^k \mathcal{P}_2^{n-k})(A)
            -
            \mathcal{M}(\mathcal{P}_2^n)(A)
            \vert
            &=
            \vert
            \mathbb{E}_{D\sim \mathcal{P}_1^k \mathcal{P}_2^{n-k}}
            \Pr\{\mathcal{M}(D)\in A\}
            -
            \mathbb{E}_{D\sim \mathcal{P}_2^{n}}
            \Pr\{\mathcal{M}(D)\in A\}
            \vert \\
            &=
            \vert
            \sum_{i=1}^k
            (
            \mathbb{E}_{D\sim \mathcal{P}_1^{\,i}\mathcal{P}_2^{\,n-i}}
            \Pr\{\mathcal{M}(D)\in A\}
            -
            \mathbb{E}_{D\sim \mathcal{P}_1^{\,i-1}\mathcal{P}_2^{\,n-i+1}}
            \Pr\{\mathcal{M}(D)\in A\}
            )
            \vert.
        \end{align*}

        Let $D'$ denote the in-neighbor of $D$ that does not contain the $i$-th element in $D$, and this index can take any arbitrary value independent of $D$. Then, almost surely, 
        \begin{align*}
        \mathbb{E}_{D\sim P_1^{i}P_2^{n-i}}
        \Pr\{\mathcal{M}(D')\in A\}
        -
        \mathbb{E}_{D\sim P_1^{i-1}P_2^{n-i+1}}
        \Pr\{\mathcal{M}(D')\in A\} =0
        \end{align*}
        
        Thus,
        \begin{align*}
        & \vert\mathcal{M}(\mathcal{P}_1^k \mathcal{P}_2^{n-k})(A)-\mathcal{M}(\mathcal{P}_2^n)(A)\vert \\
        \leq &
        \sum_{i=1}^k
        \vert
        \mathbb{E}_{D'\sim \mathcal{P}_1^{i-1}\mathcal{P}_2^{n-i}}[
        \mathbb{E}_{x_i\sim \mathcal{P}_1}[\Pr\{\mathcal{M}(D)\in A\}
        - \Pr\{\mathcal{M}(D')\in A\}] \\ & -
        \mathbb{E}_{x_i\sim \mathcal{P}_2}[\Pr\{\mathcal{M}(D)\in A\}
        - \Pr\{\mathcal{M}(D')\in A\}]]
        \vert\\
        \le &
        \sum_{i=1}^k
        \mathbb{E}_{D'\sim \mathcal{P}_1^{i-1}\mathcal{P}_2^{n-i}}
        [2(1-e^{-{\varepsilon_{D}}_i})\Vert \mathcal{P}_1-\mathcal{P}_2\Vert_{\mathrm{TV}}] =2\Vert \mathcal{P}_1-\mathcal{P}_2\Vert_{\mathrm{TV}}\sum_{i=1}^k (1-e^{-{{\varepsilon_{D}}}_i}).
        \end{align*}
        
     \end{proof}

     Now we can present the lower bound in the unbounded PDP setting in Theorem \ref{theorem:lower_bound_add/remove-one}. The proof of Theorem \ref{theorem:lower_bound_add/remove-one} follows the same steps as the proof of Theorem \ref{theorem:lower_bound} by substituting Lemma \ref{lemma:PDP_2023_TV} with Lemma \ref{lemma:add/remove-one_TV}, and we omit it here.
    \begin{theorem}
        \label{theorem:lower_bound_add/remove-one}
        Given a user universe $\mathcal{U}$ and the privacy function $\mathcal{E}$, denote $\mathfrak{N}_{\sigma}$ as the family of Gaussian distribution with variance of $\sigma^2$. For any $D \sim \mathcal{N}(\mu, \sigma^2)^n$, we use $\boldsymbol{\varepsilon}_D$ to denote the (sorted) privacy vector of $D$ that contains all $\mathcal{E}(u)$ of $D$, and for any $\mathcal{E}$-PDP mechanism $\mathcal{M}$, we have
        \begin{align*}
            \inf_{\mathcal{M}} \sup_{\mathcal{N}(\mu, \sigma^2) \in \mathfrak{N}_{\sigma}} 
            \Pr\left( \vert\mathcal{M}(D) - \mu\vert \ge \max_{k=1, 2, \hdots, n} \frac{1}{\sqrt{2}} \frac{\sigma}{\sum_{i=1}^{k}{{\varepsilon}_D}_i + 2 \sqrt{n - k}}
            \right) \geq \frac{1}{4}.
        \end{align*}
        \end{theorem}

\subsection{Statistical Mean Estimation Under Unbounded PDP Setting}
\label{section:upper_bound_unbounded}
For statistical mean estimation under the unbounded PDP setting, our key idea is to reduce this problem to a bounded PDP problem, allowing us to apply the technique developed in Section~\ref{section:upper_bound_bounded}. 
More precisely, we construct an element-wise mapping that transforms the original dataset $D$ with privacy vector $\boldsymbol{\varepsilon} (D)$ into a new dataset $\ddot{D}$ with privacy vector $\boldsymbol{\varepsilon}(\ddot{D})$. We then invoke our bounded PDP mean estimation algorithm on $\ddot{D}$ with privacy vector $\boldsymbol{\varepsilon}(\ddot{D})$.
Recall that there are two fundamental differences between the bounded PDP and unbounded PDP settings: they have different neighboring definitions, and the bounded PDP assumes a dataset with a public privacy vector, which becomes private in the unbounded setting.
Therefore, to ensure that this reduction preserves PDP, the reduction process should follow several properties:

(1) The mapped privacy vector $\boldsymbol{\varepsilon}(\ddot{D})$ should follow PDP so that we can regard it as public information in mean estimation that follows. 

(2) $\boldsymbol{\varepsilon}(\ddot{D})$ has no privacy downgrade compared with $\boldsymbol{\varepsilon}(D)$: For any element of $D$ that also appears in $\ddot{D}$, its privacy budget in $\ddot{D}$ is at most its budget in $D$.

(3) The transformation from $D$ to $\ddot{D}$ should be \textit{distance-preserving}: for any insert/delete-one neighbor $D$ and $D'$, after the transformation, $\ddot{D}$ and $\ddot{D}'$ should be constant-distance change-one neighbors, i.e., $\ddot{D}$ can be obtained from $\ddot{D}'$ by changing a constant number of elements' data.

\paragraph{Obtaining the Privacy Vector}
With the above targets in mind, now let us delve into the details. We begin by constructing a mapping from the original privacy vector to the estimated privacy vector. Here, we would like this mapping to follow several properties: 

(a) The mapping result $\boldsymbol{\varepsilon}(\ddot{D})$ preserves PDP.

(b) It should be a \emph{non-increasing mapping}, i.e., no privacy budget increases after the mapping. 

(c) We should try to do a \emph{shrinkage mapping}: The size of the mapped privacy vector should be no larger, or at least not substantially larger, than that of the original privacy vector.
 
(d) The shrinkage in (c) should not be too much: The size of the mapped privacy vector should not be substantially smaller than that of the original one.

Properties (a) and (b) are directly corresponding to ensure privacy properties (1) and (2) of reduction, whereas (c) and (d) relate to utility. For (c), if the size of the mapped privacy vector becomes large, it implies that many dummy values (such as zeros) must be inserted, which degrades the utility. For (d), obtaining the shrunk privacy vector requires deleting certain elements: If the mapped vector becomes too small, excessive user information will be discarded, resulting in a large DP error.

To satisfy the above requirements, a straightforward idea is to employ a linear domain-partitioning scheme. Recall that all privacy budgets in $D$ lie in a known range $[\varepsilon_{\min}, \varepsilon_{\max}]$, where $\varepsilon_{\max} \leq 1$. Here, we first partition this range into several sub-domains uniformly. For each sub-domain, we privatize the count of elements it contains by adding Laplace noise inversely proportional to $2^i\varepsilon_{\min}$, which is the left endpoint of that sub-domain. Then, we insert the upgraded privacy budgets $2^i\varepsilon_{\min}$, the number of which equals the privatized count, into the mapped privacy vector, thereby satisfying (a) and (b) simultaneously. To further guarantee (c), we truncate the privatized size of each sub-domain. More precisely, for each sub-domain, we subtract the privatized count by an offset term inversely proportional to $2^i\varepsilon_{\min}$ to ensure it is an underestimation of the real counting result before inserting it into the mapped privacy vector, satisfying (c).


However, this idea does not satisfy property (d). 
That is because, in such a uniform linear partitioning scheme, the number of truncated elements grows linearly with the number of sub-intervals. Therefore, if the domain is partitioned too finely, a large number of truncations will incur significant information loss. Conversely, if the domain is partitioned too coarsely, the mapped privacy budgets become overly crude and substantially underestimate the original privacy vector, subsequently bringing large DP error. In the worst case, if we only use one domain, then it is equivalent to directly applying the global minimum privacy budget $\varepsilon_{\min}$.

To address the issue, we partition the privacy domain into $\lfloor \log(\varepsilon_{\max}/\varepsilon_{\min}) \rfloor$ number of exponentially expanding sub-domains. More precisely, we partition $[\varepsilon_{\min}, \varepsilon_{\max}]$ as

\begin{align*}
     [\varepsilon_{\min}, 2\varepsilon_{\min}), [2\varepsilon_{\min}, 2^2\varepsilon_{\min}), \cdots, [2^{\lfloor \log_2(T) \rfloor}\varepsilon_{\min}, \varepsilon_{\max}]
     \label{eq:domain_partitioning}
\end{align*}
where $T := \varepsilon_{\max} / \varepsilon_{\min}$.

Then, for each sub-domain, we get a privatized counting result by adding Laplace noise and subtracting an offset term. More precisely, for the $i$-th domain, we define $D_i := D \cap_{\varepsilon} [2^i \varepsilon_{\min}, 2^{i+1} \varepsilon_{\min})$ to contain elements in the $i$-th domain only, where $\cap_{\varepsilon}$ denotes taking intersection with respect to the privacy domain. Then, we have 
\begin{align*}
    \widetilde{\texttt{Count}}(D_i) = \vert D_i \vert + \text{Lap}(\frac{1}{2^{i-1}\varepsilon_{\min}}) - \frac{1}{2^{i-1} \varepsilon_{\min}} \ln\left(\frac{\lfloor\log_2(T)\rfloor}{\beta}\right).
\end{align*}
By Laplace tail bound, we have all $ \widetilde{\texttt{Count}}(D_i)$'s to be underestimations of $\vert D_i \vert$'s with probability at least $1 - \beta$. 
Besides, we further set those $\widetilde{\texttt{Count}}(D_i) < 0$ to be zero. Finally, for each $i$-th sub-domain, we insert $\widetilde{\texttt{Count}}(D_i)$ number of $\boldsymbol{\varepsilon}_i$ into $\boldsymbol{\varepsilon}(\ddot{D})$.

It can be shown that the resulting privacy vector satisfies the above four desired properties.
First, for any user whose privacy budget lies in the range $[2^{i-1}\varepsilon_{\min},\, 2^{i}\varepsilon_{\min}]$, the user can only influence the count of the $i$-th sub-domain, and adding Laplace noise with scale $1/(2^{i-1}\varepsilon_{\min})$ is therefore sufficient to preserve his privacy. Hence, property (a) is satisfied.
Second, for each sub-domain, we insert into the mapped privacy vector an underestimation of the original privacy budgets, thereby preserving property (b).
Third, as shown earlier, each $\widetilde{\texttt{Count}}(D_i)$ is an underestimation of $\lvert D_i\rvert$ with high probability, which ensures property (c).
Finally, since there are at most $\lfloor \log_2(T)\rfloor$ sub-domains and each sub-domain deletes at most 
$O\left(\frac{1}{2^{i-1}\varepsilon_{\min}} \ln\left(\frac{\lfloor\log_2(T)\rfloor}{\beta}\right)\right)$ privacy budgets, the total number of deletions is 
$O\left( \frac{1}{\varepsilon_{\min}(D)} 
\ln\left(\frac{\lfloor\log_2(T)\rfloor}{\beta}\right)\right)$.
This establishes property (d).

\paragraph{Obtaining the Shrunk Dataset} After obtaining the mapped privacy vector, we now transform the original dataset $D$, with privacy vector $\boldsymbol{\varepsilon}(D)$, into a new dataset $\ddot{D}$ with privacy vector $\boldsymbol{\varepsilon}(\ddot{D})$ by deleting certain records and upgrading the privacy of selected users. Specifically, in the $i$-th sub-domain, we randomly sample $\widetilde{\texttt{Count}}(D_i)$ elements without replacement. If $\lvert D_i \rvert < \widetilde{\texttt{Count}}(D_i)$, the output is padded with zeros; however, this occurs with very low probability, since each $\widetilde{\texttt{Count}}(D_i)$ is, with high probability, an underestimation of $\lvert D_i \rvert$. Each sampled element is then assigned a privacy budget of $2^i \varepsilon_{\min}$. Applying this step across all sub-domains yields the shrunk dataset $\ddot{D}$. It can be shown that this transformation is distance-preserving. Since the mapped privacy vector already satisfies PDP and can be treated as public information, if $D$ and $D'$ differ by one user, then the resulting datasets $\ddot{D}$ and $\ddot{D}'$ differ by at most two elements.

\paragraph{Performing Mean Estimation on the Shrunk Dataset} With the privacy-known dataset $\ddot{D}$, we can apply Algorithm \ref{alg:PDPMeanEstimation} to obtain the estimated mean. The whole procedure is given in Algorithm \ref{alg:PDPMeanEstimationUnknownPrivacy}, in which $\texttt{WORSample}(S, m)$ denotes any algorithm that achieves sampling $m$ elements without replacement from $S$. In particular, if $m > \vert S \vert$, the output is padded with zeros so that the total number of returned samples is $m$. The whole procedure is given in Algorithm \ref{alg:PDPMeanEstimationUnknownPrivacy}.

     \begin{minipage}[ht]{0.98\linewidth}
            \raggedright
            \begin{algorithm}[H]                \caption{PDPMeanEstimationUnknownPrivacy}
                \label{alg:PDPMeanEstimationUnknownPrivacy}
                \KwIn{$D$, $\varepsilon_{\min}, \varepsilon_{\max}$, $\beta$}
                \KwOut{$\widehat{\texttt{Mean}}(D)$}
                $[\varepsilon_{\min}, 2\varepsilon_{\min}), [2\varepsilon_{\min}, 2^2\varepsilon_{\min}), \cdots, [2^{\lfloor \log_2(T) \rfloor}\varepsilon_{\min}, \varepsilon_{\max}] \gets [\varepsilon_{\min}, \varepsilon_{\max}]$
                $\ddot{\boldsymbol{\varepsilon}}, \ddot{D} \gets \emptyset$
                \For{$i \gets 1, 2, \cdots, 2^{\lfloor \log_2(T) \rfloor}$}
                {$\widehat{\texttt{Count}}(D \cap_{\varepsilon} [2^{i-1} \varepsilon_{\min}, 2^i \varepsilon_{\min})) \gets \texttt{Count}(D \cap_{\varepsilon} [2^{i-1} \varepsilon_{\min}, 2^i \varepsilon_{\min})) + \text{Lap}(\frac{1}{2^{i-1}\varepsilon_{\min}})$
                $\widetilde{\texttt{Count}}(D \cap_{\varepsilon} [2^{i-1} \varepsilon_{\min}, 2^i \varepsilon_{\min})) \gets \widehat{\texttt{Count}}(D \cap_{\varepsilon} [2^{i-1} \varepsilon_{\min}, 2^i \varepsilon_{\min})) - \frac{1}{2^{i-1} \varepsilon_{\min}} \ln\left(\frac{2\lfloor\log_2(T)\rfloor}{\beta}\right)$
                \If{$\widetilde{\emph{\texttt{Count}}}(D \cap_{\varepsilon} [2^{i-1} \varepsilon_{\min}, 2^i \varepsilon_{\min})) > 0$}{
                \tcp{Adding $\widetilde{\texttt{Count}}(D \cap_{\varepsilon} [2^{i-1} \varepsilon_{\min}, 2^i \varepsilon_{\min}))$ elements}
                $\ddot{\boldsymbol{\varepsilon}} \gets \ddot{\boldsymbol{\varepsilon}} \cup \{2^{i-1}\varepsilon_{\min}, \hdots, 2^{i-1}\varepsilon_{\min} \}$
                $\ddot{D} \gets \ddot{D} \cup \texttt{WORSample}(D \cap_{\varepsilon} [2^{i-1} \varepsilon_{\min}, 2^i \varepsilon_{\min})), \widetilde{\texttt{Count}}(D \cap_{\varepsilon} [2^{i-1} \varepsilon_{\min}, 2^i \varepsilon_{\min})))$
                }}
                $\widehat{\texttt{Mean}}(D) = \texttt{PDPMeanEstimation}(\ddot{D}, \frac{\ddot{\boldsymbol{\varepsilon}}}{4},\frac{\beta}{2})$\;
                \Return{}~$\widehat{\texttt{Mean}}(D)$
                \end{algorithm}
            \end{minipage}

     Algorithm \ref{alg:PDPMeanEstimationUnknownPrivacy} is $\mathcal{E}$-PDP and its error turns out to be near-optimal:
     
     \begin{theorem}
        \label{theorem:upper_bound_add/remove-one}
        Given a user universe $\mathcal{U}$ and the privacy function $\mathcal{E}$, for any $D \sim \mathcal{N}(\mu, \sigma^2)^n$, if \begin{align*}
        n = \Omega(\frac{1}{\varepsilon_{\min}^{3/2}(D)}\ln(\frac{\lfloor\log_2(\varepsilon_{\max} / \varepsilon_{\min})\rfloor}{\beta}) + \frac{1}{\varepsilon_{\min}(D)}\log(\frac{1}{\beta}\log(\log(\frac{\sum_{i=1}^k\varepsilon_i + (n-k)\tau)}{\beta}))\log(\frac{1}{\beta})),
        \end{align*} with probability at least $1-\beta$, Algorithm \ref{alg:PDPMeanEstimationUnknownPrivacy} preserves $\mathcal{E}$-PDP, and its error matches the lower bound $\max_{k=1, 2, \hdots, n} \frac{1}{\sqrt{2}} \frac{\sigma}{\sum_{i=1}^{k}{{\varepsilon}_D}_i + 2 \sqrt{n - k}}$ in Theorem \ref{theorem:lower_bound_add/remove-one} up to logarithmic factors.
    \end{theorem}

     \begin{proof}
    
         First, we prove the privacy of Algorithm \ref{alg:PDPMeanEstimationUnknownPrivacy}.

         The first queries that consume the privacy budget $\boldsymbol{\varepsilon}$ are the count queries. Since the domain partitioning technique ensures that each user only accounts for one domain, the counting procedure in Algorithm \ref{alg:PDPMeanEstimationUnknownPrivacy} is $\frac{\boldsymbol{\varepsilon}}{2}$-DP following the parallel composition. 
         
         Since $\widetilde{\texttt{Count}}(D_i)$ is the perturbed underestimation of $\texttt{Count}(D_i)$ (otherwise $0$ is inserted, which will not degrade privacy), after deleting elements in each sub-domain, the privacy vector of the yielded $\ddot{D}$ are public. Consider $D$ and $D'$, and the yielded shrunk datasets $\ddot{D}$ and $\ddot{D}'$, we assume $\vert D \vert - \vert D' \vert = 1$ and they differ by $u$, without loss of generality. Since the mapped privacy vector is public, we can directly use the hybrid neighboring relations (Theorem 11) in \citep{Coupling} to see that $\ddot{D}$ and $\ddot{D}'$ are at most two distance neighbors.
         Therefore, applying sequential composition shows that Algorithm \ref{alg:PDPMeanEstimationUnknownPrivacy} is $\mathcal{E}$-PDP.
         
         To analyze its utility, we break down the process as $D \rightarrow \dot{D} \rightarrow \ddot{D}$, where $D$ is the original dataset, $\dot{D}$ is $D$ with each $\varepsilon$ changed to $2^i \varepsilon_{\min}$ in each sub-domain, and $\ddot{D}$ is $\dot{D}$ with $\widetilde{\texttt{Count}}(D_i)$ samples kept in each sub-domain.

    
            First, we see the relationship between $\dot{D}$ and $\ddot{D}$. Suppose the first $t$ sub-domains are all omitted, for $i \leq t$, using Laplace tail bound, with probability at least $1 - \frac{\beta}{2\log_2(T)}$,
            \begin{align}
                \vert D_i \vert \leq 2\frac{1}{2^i \varepsilon_{\min}} \ln\left(\frac{\lfloor\log_2(T)\rfloor}{\beta}\right).
            \end{align}
            Also, for $i > t$, with probability at least $1 - \frac{\beta}{2\log_2(T)}$, 
            \begin{align}
                \texttt{Count}(D_i) - \widetilde{\texttt{Count}}(D_i) \leq \frac{1}{2^i \varepsilon_{\min}} \ln\left(\frac{\lfloor\log_2(T)\rfloor}{\beta}\right).
            \end{align}
            Therefore, by union bound, with probability at least $1 - \beta/2$, the total samples deleted from $\dot{D}$ is $O(\frac{1}{\varepsilon_{\min}(\dot{D})} \ln\left(\frac{\lfloor\log_2(T)\rfloor}{\beta}\right))$ (note that the first sub-domain involving deleting is the one that $\varepsilon_{\min}(\dot{D})$ resides). 
            
            Let us now compare the errors on $\dot{D}$ and $\ddot{D}$, note that Algorithm \ref{alg:PDPMeanEstimation} is near-optimal, we can directly see the lower bound as it has a simpler formulation.

            We compare the lower bounds on $\dot{D}$ and $\ddot{D}$: $\max_k \frac{1}{\sum_{i=1}^k\varepsilon_i + \sqrt{n-k}}$ and $\max_{k'} \frac{1}{\sum_{i=1}^{k'}\varepsilon_i' + \sqrt{n'-k'}}$, respectively, where $\vert \ddot{D}\vert$ is denoted as $n'$. We use a case analysis on $k'$:
            \begin{itemize}
                \item $n' - k' = \Omega(\frac{1}{\varepsilon_{\min}(\dot{D})}\ln(\frac{\log_2(T)}{\beta}))$:

                In this case, we set $k = k'$, i.e., we see the lower bound on $\dot{D}$ at $k = k'$, since it suffices to show the lower bounds at $k'$ match (please note that $k'$ is where the lower bound on $\ddot{D}$ takes the maximum).

                Since $k = k'$, for the first terms in the denominators of the lower bounds: $\sum_{i=1}^{k'} \varepsilon_i' \geq \sum_{i=1}^k \varepsilon_i$. For the second terms, $n - k = n - n' + n' - k'$. Since $n - n' = O(\frac{1}{\varepsilon_{\min}(\dot{D})} \ln(\frac{\log_2(T)}{\beta}))$, $\sqrt{n - k}$ and $\sqrt{n' - k'}$ are of the same order. Therefore, the lower bounds on $\ddot{D}$ and $\dot{D}$ are also of the same order.

                \item $n' - k' = o(\frac{1}{\varepsilon_{\min}(\dot{D})}\ln(\frac{\log_2(T)}{\beta}))$:

                This case suggests that $k', n', n$ are all of the order $\Omega(\frac{1}{\varepsilon_{\min}(\dot{D})}\ln(\frac{\log_2(T)}{\beta}))$. However, we have no information on what value $k$ takes to let the lower bound on $\dot{D}$ take the maximum. Set $k = k'$, we have $\sum_{i=1}^{k'} \varepsilon_i' \geq k' \varepsilon_{\min}(\dot{D}) = \Theta(n'\varepsilon_{\min}(\dot{D})) = \Theta(n\varepsilon_{\min}(\dot{D}))$. $\sqrt{n - k} = \sqrt{n - n' + n' - k'} = O(\frac{1}{\sqrt{\varepsilon_{\min}(\dot{D})}}\ln(\frac{\log_2(T)}{\beta}))$. Thus, if $n = \Omega(\frac{1}{\varepsilon_{\min}^{3/2}(\dot{D})}\ln(\frac{\log_2(T)}{\beta}))$, the lower bounds on $\ddot{D}$ and $\dot{D}$ are of the same order. 
            \end{itemize}

            Therefore, if $n = \Omega(\frac{1}{\varepsilon_{\min}^{3/2}(\dot{D})}\ln(\frac{\log_2(T)}{\beta}))$,
            \begin{align}
                \texttt{Err}(\mu_{\text{ADPM}}(\ddot{D})) = \Theta(\texttt{Err}(\mu_{\text{ADPM}}(\dot{D}))).
            \end{align}

            Next, we see the relationships between $D$ and $\dot{D}$. Recall that $\dot{D}$ is $D$ with each $\varepsilon_i$ changed to $2^i \varepsilon_{\min}$ in each sub-domain, since the ratio that each $\varepsilon_i$ changes is constrained to $[\frac{1}{2}, 2]$, the lower bound remains the same order:
            \begin{align}
                \texttt{Err}(\mu_{\text{ADPM}}(\dot{D})) = \Theta(\texttt{Err}(\mu_{\text{ADPM}}(D))).
            \end{align}
            Now we have $\texttt{Err}(\mu_{\text{ADPM}}(\ddot{D})) = \Theta(\texttt{Err}(\mu_{\text{ADPM}}(D)))$. Combining it with Theorem \ref{theorem:lower_bound_add/remove-one} completes the proof.

        \end{proof}

\bibliographystyle{plainnat}
\bibliography{bibfile}

@article{PDP_2023,
author = {Chaudhuri, Syomantak and Miagkov, Konstantin and Courtade, Thomas A.},
title = {Mean Estimation Under Heterogeneous Privacy Demands},
year = {2024},
issue_date = {Feb. 2025},
publisher = {IEEE Press},
volume = {71},
number = {2},
issn = {0018-9448},
url = {https://doi.org/10.1109/TIT.2024.3511498},
doi = {10.1109/TIT.2024.3511498},
journal = {IEEE Trans. Inf. Theor.},
month = dec,
pages = {1362–1375},
numpages = {14}
}

@inproceedings{Universal,
author = {Dong, Wei and Yi, Ke},
title = {Universal Private Estimators},
year = {2023},
isbn = {9798400701276},
publisher = {Association for Computing Machinery},
address = {New York, NY, USA},
url = {https://doi.org/10.1145/3584372.3588669},
doi = {10.1145/3584372.3588669},
booktitle = {Proceedings of the 42nd ACM SIGMOD-SIGACT-SIGAI Symposium on Principles of Database Systems},
pages = {195–206},
numpages = {12},
location = {Seattle, WA, USA},
series = {PODS '23}
}

@article{dwork2014algorithmic,
  title={The algorithmic foundations of differential privacy},
  author={Dwork, Cynthia and Roth, Aaron and others},
  journal={Foundations and Trends{\textregistered} in Theoretical Computer Science},
  volume={9},
  number={3--4},
  pages={211--407},
  year={2014},
  publisher={Now Publishers, Inc.}
}

@inproceedings{dwork2006calibrating,
  title={Calibrating noise to sensitivity in private data analysis},
  author={Dwork, Cynthia and McSherry, Frank and Nissim, Kobbi and Smith, Adam},
  booktitle={Theory of Cryptography: Third Theory of Cryptography Conference, TCC 2006, New York, NY, USA, March 4-7, 2006. Proceedings 3},
  pages={265--284},
  year={2006},
  organization={Springer}
}

@inproceedings{SVT,
author = {Dwork, Cynthia and Naor, Moni and Reingold, Omer and Rothblum, Guy N. and Vadhan, Salil},
title = {On the complexity of differentially private data release: efficient algorithms and hardness results},
year = {2009},
isbn = {9781605585062},
publisher = {Association for Computing Machinery},
address = {New York, NY, USA},
url = {https://doi.org/10.1145/1536414.1536467},
doi = {10.1145/1536414.1536467},
booktitle = {Proceedings of the Forty-First Annual ACM Symposium on Theory of Computing},
pages = {381–390},
numpages = {10},
location = {Bethesda, MD, USA},
series = {STOC '09}
}

@inproceedings{Coupling,
author = {Balle, Borja and Barthe, Gilles and Gaboardi, Marco},
title = {Privacy amplification by subsampling: tight analyses via couplings and divergences},
year = {2018},
publisher = {Curran Associates Inc.},
address = {Red Hook, NY, USA},
booktitle = {Proceedings of the 32nd International Conference on Neural Information Processing Systems},
pages = {6280–6290},
numpages = {11},
location = {Montr\'{e}al, Canada},
series = {NIPS'18}
}

@InProceedings{BartheOlmedo2013,
author="Barthe, Gilles
and Olmedo, Federico",
editor="Fomin, Fedor V.
and Freivalds, R{\={u}}si{\c{n}}{\v{s}}
and Kwiatkowska, Marta
and Peleg, David",
title="Beyond Differential Privacy: Composition Theorems and Relational Logic for f-divergences between Probabilistic Programs",
booktitle="Automata, Languages, and Programming",
year="2013",
publisher="Springer Berlin Heidelberg",
address="Berlin, Heidelberg",
pages="49--60",
isbn="978-3-642-39212-2"
}

@article{mcdiarmid1989method,
  title={On the method of bounded differences},
  author={McDiarmid, Colin and others},
  journal={Surveys in combinatorics},
  volume={141},
  number={1},
  pages={148--188},
  year={1989},
  publisher={Norwich}
}

@book{Concentration_Inequalities,
    author = {Boucheron, Stéphane and Lugosi, Gábor and Massart, Pascal},
    title = {Concentration Inequalities: A Nonasymptotic Theory of Independence},
    publisher = {Oxford University Press},
    year = {2013},
    month = {02},
    isbn = {9780199535255},
    doi = {10.1093/acprof:oso/9780199535255.001.0001},
    url = {https://doi.org/10.1093/acprof:oso/9780199535255.001.0001},
}

@article{maurer2009empirical,
  title={Empirical bernstein bounds and sample variance penalization},
  author={Maurer, Andreas and Pontil, Massimiliano},
  journal={arXiv preprint arXiv:0907.3740},
  year={2009}
}

@article{Duchi02012018,
author = {John C. Duchi and Michael I. Jordan and Martin J. Wainwright and},
title = {Minimax Optimal Procedures for Locally Private Estimation},
journal = {Journal of the American Statistical Association},
volume = {113},
number = {521},
pages = {182--201},
year = {2018},
publisher = {ASA Website},
doi = {10.1080/01621459.2017.1389735},


URL = { 
    
        https://doi.org/10.1080/01621459.2017.1389735
    
    

},
eprint = { 
    
        https://doi.org/10.1080/01621459.2017.1389735
}

}

@article{Parallel_Composition,
author = {McSherry, Frank},
title = {Privacy integrated queries: an extensible platform for privacy-preserving data analysis},
year = {2010},
issue_date = {September 2010},
publisher = {Association for Computing Machinery},
address = {New York, NY, USA},
volume = {53},
number = {9},
issn = {0001-0782},
url = {https://doi.org/10.1145/1810891.1810916},
doi = {10.1145/1810891.1810916},
journal = {Commun. ACM},
month = sep,
pages = {89–97},
numpages = {9}
}

@article{Markov_Kernel,
author = {Belavkin, Roman V.},
title = {Optimal measures and Markov transition kernels},
year = {2013},
issue_date = {February  2013},
publisher = {Kluwer Academic Publishers},
address = {USA},
volume = {55},
number = {2},
issn = {0925-5001},
url = {https://doi.org/10.1007/s10898-012-9851-1},
doi = {10.1007/s10898-012-9851-1},
journal = {J. of Global Optimization},
month = feb,
pages = {387–416},
numpages = {30}
}

@book{RN_derivative,
  title={Handbook of measure Theory: In two volumes},
  author={Pap, Endre},
  year={2002},
  publisher={Elsevier}
}

@INPROCEEDINGS{Graham,
  author={Jorgensen, Zach and Yu, Ting and Cormode, Graham},
  booktitle={2015 IEEE 31st International Conference on Data Engineering}, 
  title={Conservative or liberal? Personalized differential privacy}, 
  year={2015},
  volume={},
  number={},
  pages={1023-1034},
  doi={10.1109/ICDE.2015.7113353}}

@inproceedings{add_remove_one_mean_estimation,
author = {Kulesza, Alex and Suresh, Ananda Theertha and Wang, Yuyan},
title = {Mean estimation in the add-remove model of differential privacy},
year = {2024},
publisher = {JMLR.org},
booktitle = {Proceedings of the 41st International Conference on Machine Learning},
articleno = {1027},
numpages = {20},
location = {Vienna, Austria},
series = {ICML'24}
}

@inproceedings{clipping_1,
author = {Huang, Ziyue and Liang, Yuting and Yi, Ke},
title = {Instance-optimal mean estimation under differential privacy},
year = {2021},
isbn = {9781713845393},
publisher = {Curran Associates Inc.},
address = {Red Hook, NY, USA},
booktitle = {Proceedings of the 35th International Conference on Neural Information Processing Systems},
articleno = {1990},
numpages = {12},
series = {NIPS '21}
}

@inproceedings{clipping_2,
 author = {Biswas, Sourav and Dong, Yihe and Kamath, Gautam and Ullman, Jonathan},
 booktitle = {Advances in Neural Information Processing Systems},
 editor = {H. Larochelle and M. Ranzato and R. Hadsell and M.F. Balcan and H. Lin},
 pages = {14475--14485},
 publisher = {Curran Associates, Inc.},
 title = {CoinPress: Practical Private Mean and Covariance Estimation},
 url = {https://proceedings.neurips.cc/paper_files/paper/2020/file/a684eceee76fc522773286a895bc8436-Paper.pdf},
 volume = {33},
 year = {2020}
}

@InProceedings{clipping_3,
  title = 	 {Private Mean Estimation of Heavy-Tailed Distributions},
  author =       {Kamath, Gautam and Singhal, Vikrant and Ullman, Jonathan},
  booktitle = 	 {Proceedings of Thirty Third Conference on Learning Theory},
  pages = 	 {2204--2235},
  year = 	 {2020},
  editor = 	 {Abernethy, Jacob and Agarwal, Shivani},
  volume = 	 {125},
  series = 	 {Proceedings of Machine Learning Research},
  month = 	 {09--12 Jul},
  publisher =    {PMLR},
  url = 	 {https://proceedings.mlr.press/v125/kamath20a.html}
}

@InProceedings{partitioning_based_PDP,
author="Li, Haoran
and Xiong, Li
and Ji, Zhanglong
and Jiang, Xiaoqian",
editor="Kim, Jinho
and Shim, Kyuseok
and Cao, Longbing
and Lee, Jae-Gil
and Lin, Xuemin
and Moon, Yang-Sae",
title="Partitioning-Based Mechanisms Under Personalized Differential Privacy",
booktitle="Advances in Knowledge Discovery and Data Mining",
year="2017",
publisher="Springer International Publishing",
address="Cham",
pages="615--627",
isbn="978-3-319-57454-7"
}

@INPROCEEDINGS {Privacy_Aware_Agents,
author = { Cummings, Rachel and Elzayn, Hadi and Pountourakis, Emmanouil and Gkatzelis, Vasilis and Ziani, Juba },
booktitle = { 2023 IEEE Conference on Secure and Trustworthy Machine Learning (SaTML) },
title = {{ Optimal Data Acquisition with Privacy-Aware Agents }},
year = {2023},
volume = {},
ISSN = {},
pages = {210-224},
doi = {10.1109/SaTML54575.2023.00023},
url = {https://doi.ieeecomputersociety.org/10.1109/SaTML54575.2023.00023},
publisher = {IEEE Computer Society},
address = {Los Alamitos, CA, USA},
month =Feb}

@inproceedings{PDP_platform,
author = {Fallah, Alireza and Makhdoumi, Ali and Malekian, Azarakhsh and Ozdaglar, Asuman},
title = {Optimal and Differentially Private Data Acquisition: Central and Local Mechanisms},
year = {2022},
isbn = {9781450391504},
publisher = {Association for Computing Machinery},
address = {New York, NY, USA},
url = {https://doi.org/10.1145/3490486.3538329},
doi = {10.1145/3490486.3538329},
booktitle = {Proceedings of the 23rd ACM Conference on Economics and Computation},
pages = {1141},
numpages = {1},
location = {Boulder, CO, USA},
series = {EC '22}
}

@article{Alaggan,
  title={Heterogeneous differential privacy},
  author={Alaggan, Mohammad and Gambs, S{\'e}bastien and Kermarrec, Anne-Marie},
  journal={arXiv preprint arXiv:1504.06998},
  year={2015}
}

@inproceedings{KV18,
  title={Finite Sample Differentially Private Confidence Intervals},
  author={Karwa, Vishesh and Vadhan, Salil},
  booktitle={9th Innovations in Theoretical Computer Science Conference (ITCS 2018)},
  pages={44--1},
  year={2018},
  organization={Schloss Dagstuhl--Leibniz-Zentrum f{\"u}r Informatik}
}

@inproceedings{Alabi_stoc,
author = {Alabi, Daniel and Kothari, Pravesh K. and Tankala, Pranay and Venkat, Prayaag and Zhang, Fred},
title = {Privately Estimating a Gaussian: Efficient, Robust, and Optimal},
year = {2023},
isbn = {9781450399135},
publisher = {Association for Computing Machinery},
address = {New York, NY, USA},
url = {https://doi.org/10.1145/3564246.3585194},
doi = {10.1145/3564246.3585194},
booktitle = {Proceedings of the 55th Annual ACM Symposium on Theory of Computing},
pages = {483–496},
numpages = {14},
location = {Orlando, FL, USA},
series = {STOC 2023}
}

@InProceedings{colt_2019,
  title = 	 {Privately Learning High-Dimensional Distributions},
  author =       {Kamath, Gautam and Li, Jerry and Singhal, Vikrant and Ullman, Jonathan},
  booktitle = 	 {Proceedings of the Thirty-Second Conference on Learning Theory},
  pages = 	 {1853--1902},
  year = 	 {2019},
  editor = 	 {Beygelzimer, Alina and Hsu, Daniel},
  volume = 	 {99},
  series = 	 {Proceedings of Machine Learning Research},
  month = 	 {25--28 Jun},
  publisher =    {PMLR},
  url = 	 {https://proceedings.mlr.press/v99/kamath19a.html}
}

@article{Private_estimation_with_public_data,
  title={Private estimation with public data},
  author={Bie, Alex and Kamath, Gautam and Singhal, Vikrant},
  journal={Advances in neural information processing systems},
  volume={35},
  pages={18653--18666},
  year={2022}
}

@inproceedings{heavy_tailed,
  title={Improved rates for differentially private stochastic convex optimization with heavy-tailed data},
  author={Kamath, Gautam and Liu, Xingtu and Zhang, Huanyu},
  booktitle={International Conference on Machine Learning},
  pages={10633--10660},
  year={2022},
  organization={PMLR}
}

@inproceedings{colt_2023_fast,
  title={A pretty fast algorithm for adaptive private mean estimation},
  author={Kuditipudi, Rohith and Duchi, John and Haque, Saminul},
  booktitle={The Thirty Sixth Annual Conference on Learning Theory},
  pages={2511--2551},
  year={2023},
  organization={PMLR}
}

@article{Average_case_averages,
  title={Average-case averages: Private algorithms for smooth sensitivity and mean estimation},
  author={Bun, Mark and Steinke, Thomas},
  journal={Advances in Neural Information Processing Systems},
  volume={32},
  year={2019}
}

@inproceedings{aden2021sample,
  title={On the sample complexity of privately learning unbounded high-dimensional gaussians},
  author={Aden-Ali, Ishaq and Ashtiani, Hassan and Kamath, Gautam},
  booktitle={Algorithmic Learning Theory},
  pages={185--216},
  year={2021},
  organization={PMLR}
}

@inproceedings{ashtiani2022private,
  title={Private and polynomial time algorithms for learning gaussians and beyond},
  author={Ashtiani, Hassan and Liaw, Christopher},
  booktitle={Conference on Learning Theory},
  pages={1075--1076},
  year={2022},
  organization={PMLR}
}

@inproceedings{kothari2022private,
  title={Private robust estimation by stabilizing convex relaxations},
  author={Kothari, Pravesh and Manurangsi, Pasin and Velingker, Ameya},
  booktitle={Conference on Learning Theory},
  pages={723--777},
  year={2022},
  organization={PMLR}
}

@inproceedings{bun2014fingerprinting,
  title={Fingerprinting codes and the price of approximate differential privacy},
  author={Bun, Mark and Ullman, Jonathan and Vadhan, Salil},
  booktitle={Proceedings of the forty-sixth annual ACM symposium on Theory of computing},
  pages={1--10},
  year={2014}
}

@inproceedings{dwork2015robust,
  title={Robust traceability from trace amounts},
  author={Dwork, Cynthia and Smith, Adam and Steinke, Thomas and Ullman, Jonathan and Vadhan, Salil},
  booktitle={2015 IEEE 56th Annual Symposium on Foundations of Computer Science},
  pages={650--669},
  year={2015},
  organization={IEEE}
}

@article{AOS_2021,
  title={The cost of privacy: Optimal rates of convergence for parameter estimation with differential privacy},
  author={Cai, T Tony and Wang, Yichen and Zhang, Linjun},
  journal={The Annals of Statistics},
  volume={49},
  number={5},
  pages={2825--2850},
  year={2021},
  publisher={Institute of Mathematical Statistics}
}

@article{nips_2019,
  title={Differentially private algorithms for learning mixtures of separated gaussians},
  author={Kamath, Gautam and Sheffet, Or and Singhal, Vikrant and Ullman, Jonathan},
  journal={Advances in Neural Information Processing Systems},
  volume={32},
  year={2019}
}

@inproceedings{efficient,
author = {Hopkins, Samuel B. and Kamath, Gautam and Majid, Mahbod},
title = {Efficient mean estimation with pure differential privacy via a sum-of-squares exponential mechanism},
year = {2022},
isbn = {9781450392648},
publisher = {Association for Computing Machinery},
address = {New York, NY, USA},
url = {https://doi.org/10.1145/3519935.3519947},
doi = {10.1145/3519935.3519947},
booktitle = {Proceedings of the 54th Annual ACM SIGACT Symposium on Theory of Computing},
pages = {1406–1417},
numpages = {12},
location = {Rome, Italy},
series = {STOC 2022}
}

@inproceedings{Dwork_stoc_2014,
  title={Analyze gauss: optimal bounds for privacy-preserving principal component analysis},
  author={Dwork, Cynthia and Talwar, Kunal and Thakurta, Abhradeep and Zhang, Li},
  booktitle={Proceedings of the forty-sixth annual ACM symposium on Theory of computing},
  pages={11--20},
  year={2014}
}

@article{Hardt_nips_2014,
  title={The noisy power method: A meta algorithm with applications},
  author={Hardt, Moritz and Price, Eric},
  journal={Advances in neural information processing systems},
  volume={27},
  year={2014}
}

@article{Amin_nips_2019,
  title={Differentially private covariance estimation},
  author={Amin, Kareem and Dick, Travis and Kulesza, Alex and Munoz, Andres and Vassilvitskii, Sergei},
  journal={Advances in Neural Information Processing Systems},
  volume={32},
  year={2019}
}

@article{zhang2019,
  title={Probabilistic matrix factorization with personalized differential privacy},
  author={Zhang, Shun and Liu, Laixiang and Chen, Zhili and Zhong, Hong},
  journal={Knowledge-Based Systems},
  volume={183},
  pages={104864},
  year={2019},
  publisher={Elsevier}
}

@inproceedings{PODS_2003,
  title={Revealing information while preserving privacy},
  author={Dinur, Irit and Nissim, Kobbi},
  booktitle={Proceedings of the twenty-second ACM SIGMOD-SIGACT-SIGART symposium on Principles of database systems},
  pages={202--210},
  year={2003}
}

@article{homer2008resolving,
  title={Resolving individuals contributing trace amounts of DNA to highly complex mixtures using high-density SNP genotyping microarrays},
  author={Homer, Nils and Szelinger, Szabolcs and Redman, Margot and Duggan, David and Tembe, Waibhav and Muehling, Jill and Pearson, John V and Stephan, Dietrich A and Nelson, Stanley F and Craig, David W},
  journal={PLoS genetics},
  volume={4},
  number={8},
  pages={e1000167},
  year={2008},
  publisher={Public Library of Science San Francisco, USA}
}

@InProceedings{pmlr-v162-asi22b,
  title = 	 {Optimal Algorithms for Mean Estimation under Local Differential Privacy},
  author =       {Asi, Hilal and Feldman, Vitaly and Talwar, Kunal},
  booktitle = 	 {Proceedings of the 39th International Conference on Machine Learning},
  pages = 	 {1046--1056},
  year = 	 {2022},
  editor = 	 {Chaudhuri, Kamalika and Jegelka, Stefanie and Song, Le and Szepesvari, Csaba and Niu, Gang and Sabato, Sivan},
  volume = 	 {162},
  series = 	 {Proceedings of Machine Learning Research},
  month = 	 {17--23 Jul},
  publisher =    {PMLR},
  url = 	 {https://proceedings.mlr.press/v162/asi22b.html}
}

@inproceedings{dwork2006our,
  title={Our data, ourselves: Privacy via distributed noise generation},
  author={Dwork, Cynthia and Kenthapadi, Krishnaram and McSherry, Frank and Mironov, Ilya and Naor, Moni},
  booktitle={Annual international conference on the theory and applications of cryptographic techniques},
  pages={486--503},
  year={2006},
  organization={Springer}
}

@article{2025_statistical_estimation,
  title={A bias-accuracy-privacy trilemma for statistical estimation},
  author={Kamath, Gautam and Mouzakis, Argyris and Regehr, Matthew and Singhal, Vikrant and Steinke, Thomas and Ullman, Jonathan},
  journal={Journal of the American Statistical Association},
  pages={1--12},
  year={2025},
  publisher={Taylor \& Francis}
}
\end{document}